\documentclass[sn-mathphys,Numbered]{sn-jnl}

\usepackage{graphicx}%
\usepackage{multirow}%
\usepackage{amsmath,amssymb,amsfonts}%
\usepackage{amsthm}%
\usepackage{mathrsfs}%
\usepackage[title]{appendix}%
\usepackage{xcolor}%
\usepackage{textcomp}%
\usepackage{manyfoot}%
\usepackage{booktabs}%
\usepackage{listings}%
\usepackage{subfiles}

\usepackage{mathtools}
\usepackage{float}
\usepackage[procnumbered, linesnumbered, ruled, vlined]{algorithm2e}
\usepackage{longtable}
\usepackage{verbatim}
\usepackage{hyperref}
\usepackage{amsthm}

\theoremstyle{thmstyleone}%
\newtheorem{theorem}{Theorem}
\newtheorem{lemma}{Lemma}%
\newtheorem{corollary}{Corollary}[theorem]

\theoremstyle{thmstyletwo}%

\theoremstyle{thmstylethree}%
\newtheorem{assumption}{Assumption}%

\newcounter{experiment}
\newcommand{\R}{\mathbb{R}}
\newcommand{\E}{\mathbb{E}}
\newcommand{\N}{\mathcal{N}}

\usepackage{xcolor}

\begin{document}

\title[Article Title]{DePAint: A Decentralized Safe Multi-Agent Reinforcement Learning Algorithm considering  Peak and Average Constraints}

\author[1]{\fnm{Raheeb} \sur{Hassan}}\email{raheeb-2017915005@cs.du.ac.bd}
\equalcont{These authors contributed equally to this work.}

\author[2]{\fnm{K.M. Shadman} \sur{Wadith}}\email{kmshadman-2017714990@cs.du.ac.bd}
\equalcont{These authors contributed equally to this work.}
\author[1]{\fnm{Md. Mamun or} \sur{Rashid}}\email{mamun@cse.du.ac.bd}
\author[1]{\fnm{Md. Mosaddek} \sur{Khan}}\email{mosaddek@du.ac.bd}

\affil*[1]{\orgdiv{Department of Computer Science and Engineering}, \orgname{University of Dhaka}, \orgaddress{\city{Dhaka - 1000}, \country{Bangladesh}}}

\abstract{
 The domain of safe multi-agent reinforcement learning (MARL), despite its potential applications in areas ranging from drone delivery and vehicle automation to the development of zero-energy communities, remains relatively unexplored. The primary challenge involves training agents to learn optimal policies that maximize rewards while adhering to stringent safety constraints, all without the oversight of a central controller. These constraints are critical in a wide array of applications, including collision avoidance in autonomous systems, adherence to traffic and resource management laws in drone delivery, vehicle automation, and ensuring energy efficiency and sustainability in zero-energy communities. Moreover, ensuring the privacy of sensitive information in decentralized settings introduces an additional layer of complexity, necessitating innovative solutions that uphold privacy while achieving the system’s safety and efficiency goals.  In this paper, we address the problem of multi-agent policy optimization in a decentralized setting, where agents communicate with their neighbors to maximize the sum of their cumulative rewards while also satisfying each agent's safety constraints. We consider both peak and average constraints. In this scenario, there is no central controller coordinating the agents and both the rewards and constraints are only known to each agent locally/privately. We formulate the problem as a decentralized constrained multi-agent Markov Decision Problem and propose a momentum-based decentralized policy gradient method, DePAint, to solve it. To the best of our knowledge, this is the first privacy-preserving fully decentralized multi-agent reinforcement learning algorithm that considers both peak and average constraints.  We then provide theoretical analysis and empirical evaluation of our algorithm in a number of scenarios and compare its performance to centralized algorithms that consider similar constraints. 
}




\maketitle

\section{Introduction}
\label{chap:introduction}

The prioritization of safety has become increasingly important as artificial intelligence (AI) technologies are deployed in critical domains with the potential for significant effects \cite{amodei2016concrete}. Ensuring the safety of AI agents, particularly in complex and dynamic scenarios, poses significant challenges. In this context, safety encompasses the design, development, and deployment of AI systems that minimize the potential for negative outcomes. This involves addressing issues such as bias and impartiality, transparency, robustness, and the capacity to deal with unforeseen circumstances. Over the years, we have witnessed a continuous pursuit of research efforts aimed at addressing safety problems and, as such, enabling the dependable deployment of intelligent agents in real-world scenarios.\smallskip

In recent years, Multi-Agent Reinforcement Learning (MARL) has emerged as a promising avenue for tackling various real-world challenges. The integration of safety considerations with MARL stands to enhance the applicability of intelligent agents in diverse domains. For instance, within the realm of autonomous transportation, multi-agent reinforcement learning facilitates collaborative navigation in intricate traffic scenarios, thereby fostering safer and more efficient transportation systems. The individual agents will be able to make smarter decisions by incorporating safety constraints into MARL algorithms, thereby promoting cooperation and collision avoidance among multiple vehicles \cite{shalev2016safe}. Similarly, in the context of zero-energy communities, multi-agent reinforcement learning holds the potential to optimize energy consumption and distribution among different households, thereby advancing sustainability and minimizing wastage. Furthermore, in smart grid systems, MARL algorithms can optimize energy distribution while taking safety considerations into account \cite{alqahtani2022dynamic}. By addressing safety concerns through multi-agent reinforcement learning, these applications, among others such as multi-agent pathfinding and swarm intelligent systems, stand to become more reliable, scalable, and adaptable, paving the way for the widespread adoption and integration of intelligent agents in real-world scenarios.\smallskip

When it comes to safety in MARL, Constrained Markov Decision Processes (CMDPs) offer an augmentation to conventional Markov Decision Processes (MDPs) by integrating safety constraints into the decision-making framework \cite{altman1995constrained}. In the CMDP framework, agents are not solely focused on optimizing their expected rewards but are also bound to operate within a set of established safety constraints. This modification holds substantial relevance in various AI applications. By ensuring that agents interact securely within defined boundaries, CMDPs serve as a preventive measure against unsafe states or behaviors, thereby enhancing system resilience and alignment with safety norms. The incorporation of CMDP methodologies provides formal safety validations and addresses critical safety challenges that might be overlooked due to the inherent bias toward performance over safety in many conventional MARL algorithms.\smallskip

Over the past few years, a multitude of research efforts have sought to provide solutions to CMDP challenges. CMDP solvers are indispensable for addressing complex issues centered around constraint-driven decision-making. Constrained Policy Optimization (CPO) has emerged as a promising approach for addressing the safety and stability concerns of reinforcement learning \cite{achiam2017constrained}. Through integrating trust region methods, CPO ensures that policy updates remain confined within a region where constraints are satisfied. This mechanism offers robustness against unintended violations of safety conditions, thereby enhancing the reliability and stability of learned policies. Extending CPO, Multi-Agent Constrained Policy Optimization (MACPO) \cite{gu2021multi} jointly optimizes policies in a centralized manner while adhering to collective constraints. This collaborative optimization approach marks a significant advancement in enabling multi-agent systems to operate within safety bounds. However, coordinating policy updates and constraint representations among multiple agents introduces additional complexities. Moreover, the scalability of this centralized algorithm in larger multi-agent systems could become a bottleneck, limiting its applicability.\smallskip

To address this scalability issue in a centralized setting, there are some innovative approaches \cite{gronauer2021multi} that combine centralized training and decentralized execution. It was first introduced in the Multi-Agent Deep Deterministic Policy Gradient Algorithm (MADDPG) \cite{lowe2017multi} and was later extended using an actor-critic approach \cite{parnika2021attention}. This approach allows agents to share a centralized critique for value estimation while independently learning decentralized policies. However, as the number of agents increases, the complexity of the central critic's architecture grows significantly, potentially leading to slower training times and higher memory requirements. To overcome this problem, Safe Dec-PG \cite{lu2021decentralized} introduced the decentralized training approach. Though it overcomes the issue of scalability, the policy gradient approach suffers from high variance due to the non-convex nature of the reward function, especially in the collaborative multi-agent setting, leading to inconsistent gradient estimates that hinder stable policy updates.\smallskip

 In addition, to solve CMDP problems, most works deal with only one constraint, which is either a peak constraint \cite{bai2020provably, gattami2019reinforcement, geibel2006reinforcement, geibel2005risk} or an average constraint \cite{chow2018lyapunov, ding2021provably, gattami2019reinforcement, prashanth2016variance}. However, solving a problem with both peak and average constraints is sometimes more applicable to real-world scenarios. For example, a scenario where an autonomous drone delivery service is deployed in a bustling metropolitan city. Incorporating both peak and average constraints can ensure the safety and efficiency of drone delivery operations. The peak constraints come into play during critical moments of flight, such as when navigating through dense urban areas or landing in busy hospital helipads. The drones must adhere to strict speed and altitude limits to avoid potential collisions with other air traffic or pedestrians on the ground. The average constraint ensures that the drones maintain sustainable energy consumption and operational efficiency throughout their delivery missions. By balancing energy use and optimizing flight paths, drones can consistently deliver medical supplies while minimizing environmental impact and operational costs.\smallskip

In a multi-agent setting, the problem of addressing both peak and average constraints has been tackled by CMIX \cite{liu2021cmix}. Drawing upon the concept of QMIX \cite{rashid2018qmix} and employing value function factorization, CMIX proposed its Centralized Learning and Decentralized Execution (CTDE) algorithm to solve the multi-agent Markov Decision Problem under both peak and average constraints. However, their algorithm is centralized in nature. Recent research in a similar vein introduced a decentralized Constrained Q-Learning algorithm \cite{geng2023reinforcement}, applied in the context of Vehicular Network Routing. In our research, the primary objective is to devise a more practical MARL algorithm by incorporating decentralization and privacy. We propose a novel decentralized algorithm for networked agents, aiming to maximize the discounted cumulative reward while respecting both peak and average constraints, where the individual reward function and constraints are private and local to the respective agents. To address this challenge, we formulate the problem as a Decentralized Constrained Multi-Agent Markov Decision Problem and introduce a momentum-based Decentralized Policy Gradient method, DePAint, to solve it. Empirical results demonstrate DePAint's robust performance in diverse connected graphs, even with minimal connection networks, and its superiority over centralized algorithms dealing with similar constraints. Our research provides an efficient and safety-aware decentralized approach, showcasing the practicality and efficacy of decentralized MARL methods for real-world applications.  The contributions of this paper are summarized as follows:

\begin{itemize}
    \item We introduce a decentralized MARL algorithm, marking the first instance of such an algorithm to incorporate both peak and average constraints. This unique feature allows us to better model the constraints in real-world scenarios, offering a comprehensive solution for ensuring agent safety during training and operation.

    \item The decentralized nature of the algorithm eliminates the need for a central controller during the training process, showcasing its autonomy and scalability. This autonomy allows the algorithm to adapt to an increasing number of agents, making it a versatile and efficient solution for large-scale multi-agent systems.

    \item To address the challenge of variance associated with constrained multi-agent policy gradient methods, we have incorporated momentum-based variance reduction. This addition stabilizes the learning process, significantly reducing variance and enhancing the overall efficiency of the decentralized algorithm. The result is a more robust training framework for multi-agent systems.

    \item Our algorithm prioritizes privacy by default, introducing a mechanism that maintains the confidentiality of individual agent rewards and constraints. In the absence of a central controller, this privacy-preserving approach ensures that each local agent's reward and constraint information remains private, contributing to enhanced privacy in decentralized multi-agent systems.
\end{itemize}

We discuss the formulation of the problem and necessary preliminary concepts in the following section. We explain our proposed algorithm in Section \ref{chap:proposedMethod}. Then, in Section \ref{chap:theoreticalProof}, we theoretically analyze our proposed algorithm. We report our empirical results in Section \ref{sec:empiricalEval}, and Section \ref{sec:conclusion} concludes our discussion.

\section{Problem Formulation \& Background}
\label{chap:problemFormulationAndBackground}
In this section, we first mathematically formulate the problem we are trying to solve in Section~\ref{sec:problemFormulation}. We then discuss some preliminary concepts that are required to properly understand our research in Section~\ref{sec:background}.

\subsection{Problem Formulation}
\label{sec:problemFormulation}

Consider a team of agents, denoted by $\N = \{1, ..., n\}$. Each agent operates autonomously but can communicate with its peers via a communication network. This network is aptly represented by a well-connected, undirected graph  $G$, where the nodes and edges symbolize agents and communication links, respectively. Here, we use the term $\Delta$ so that $\Delta(\diamond)$ denotes the set of all possible probability distributions over the set $\diamond$.

Now, we formulate the decentralized safe MARL problem as an 8-tuple (i.e. $<S,\{A_i\}_{i=1}^n, P, \{R_i\}_{i=1}^n, \{C_i\}_{i=1}^n, \{K_i\}_{i=1}^n, G, \gamma>$) constrained multi-agent Markov game, where

\begin{itemize}
  \item $S$ is the global state space shared by all agents,
  \item $A:= A_1 \times ... \times A_n$ is the joint action space of all agents,
  \item $P: S \times A \rightarrow \Delta(S)$ is the state transition model,
  \item $R_i: S \times A \rightarrow \R$ is the reward function of agent i,
  \item $C_i: S \times A \rightarrow \R^{m_1}$ is the long term utility function of agent i,
  \item $K_i: S \times A \rightarrow \R^{m_2}$ is the instantaneous utility function of agent i,
  \item $\gamma$ is the discount factor,
  \item $G = (\N, \mathcal{E})$ is the graph representing the communication network, where $\N$ is the set of nodes and $\mathcal{E}$ is the set of edges.
\end{itemize}

At any point in time $t$, the collective state of all agents can be described by a state $s^t$ from the state space $S$. Each agent $i$ can take any action $a_i^t$ in $A_i$. Now, to take an action, each agent $i$ is equipped with a local policy function $\pi_i: S \rightarrow \Delta(A)$. Within this setting, Equation~\ref{eq:joint_policy} defines the joint policy, $\pi(a^t | s^t)$. So, for every state $s^t$, the corresponding joint action $a^t =$ $($ $a_1^t$, $a_2^t$, $...$, $a_n^t$ $)$ is sampled from the probability distribution given by $\pi(s^t)$.
\begin{equation}
    \pi(a^t | s^t) := \prod_{i = 1}^n \pi_{i}(a_i^t | s^t) \label{eq:joint_policy}
\end{equation}

The agents reside in an environment defined solely by the state transition model $P: S \times A \rightarrow \Delta(S)$. Given the current state $s^t$ and $a^t$, $P(s^t, a^t)$ gives a probability distribution over $S$, and then the next state $s^{t+1}$ is sampled from this probability distribution. For each state-action pair, the environment also emits three values $R_i(s^t, a^t)$, $C_i(s^t, a^t)$, and $K_i(s^t, a^t)$, which we are going to discuss shortly.\smallskip

Let $\tau =$ $($ $s^0,$ $a^0,$ $s^1,$ $...\ ,$ $s^{t-1},$ $a^{t-1},$ $s^t,$ $ \ldots\ )$ be an infinite sequence of state-action pairs. Here, each action \( a^t \) is derived from a sampling process based on the joint policy function \( \pi(s^t) \). Similarly, each subsequent state \( s^t \) is determined by sampling from the transition probability function \( P(s^{t-1}, a^{t-1}) \). Such an infinite sequence is denominated as a trajectory. Therefore, given an initial state distribution of $\rho(s^0)$, the probability of a certain trajectory $\tau$ for a given policy $\pi$ is defined in Equation~\ref{eq:traj_prob}. Based on Equation~\ref{eq:traj_prob}, we can define $p(\text{.}|\pi)$ or simply $p(\pi)$ as the probability distribution of all possible trajectories given $\pi$.
\begin{equation}
    p(\tau | \pi) := \rho(s^0) \prod_{t=0}^\infty \pi(a^t | s^t) P(s^{t+1}|s^t, a^t) \label{eq:traj_prob}
\end{equation}

For each agent \( i \), Equations~\ref{eq:disc_reward} and \ref{eq:disc_util} define two distinct terms \( J_i^{\,R} \)  and \( J_i^{\,C} \), respectively ($\underset{\tau \sim p(.|\pi)}{\E}$ denotes the expected value over all trajectories $\tau$ sampled from $p(.|\pi)$). \( J_i^{\,R} \) represents the discounted long-term return associated with the reward function, while \( J_i^{\,C} \) denotes the accumulated long-term utility derived from the utility function. These terms encapsulate the agent's prospective outcomes in terms of rewards and utilities over an extended time horizon. The aim of the agents is to find the policy $\pi$ that maximizes $J_i^{R}(\pi)$ while satisfying constraints on $J_i^{C}(\pi)$ and $K_i$. Formally, we can define the problem as defined in Equation~\ref{eq:prob}, \ref{eq:avg_const} and \ref{eq:peak_const}. Within this setting, $c_i \in \R^{m_1}$ and $k_i \in \R^{m_2}$ are the constraints imposed on those two functions. We term the constraint defined in Equation~\ref{eq:avg_const} as the \textit{average constraint} and the constraint in Equation~\ref{eq:peak_const} as the \textit{peak constraint}.  In Figure~\ref{fig:problemFormulation}, we present a pictorial representation of our problem, illustrating our objective to maximize rewards while concurrently satisfying the constraints. 

\begin{align}
    J_i^{R}(\pi) := \underset{\tau \sim p(.|\pi)}{\E} \left[\sum_{t = 0}^{\infty} \gamma^t R_i (s^t, a^t) \right]
    \label{eq:disc_reward}\\
    J_i^{C}(\pi) := \underset{\tau \sim p(.|\pi)}{\E} \left[\sum_{t = 0}^{\infty} \gamma^t C_i (s^t, a^t) \right]
    \label{eq:disc_util}
\end{align}

\begin{align}
    & \underset{\pi \in \Pi}{\max}\ \frac{1}{n} \sum_{i=0}^n J_i^{R}(\pi) \label{eq:prob}\\
    \text{subject to }\ & J_i^{C}(\pi) \ge c_i & \forall i \in \N \label{eq:avg_const}\\
    \text{and }\ & K_i(s^t, a^t) \ge k_i & \forall i \in \N, (s^t, a^t) \in p(.|\pi). \label{eq:peak_const}
\end{align}

\begin{figure}[H]
    \includegraphics[width=\textwidth]{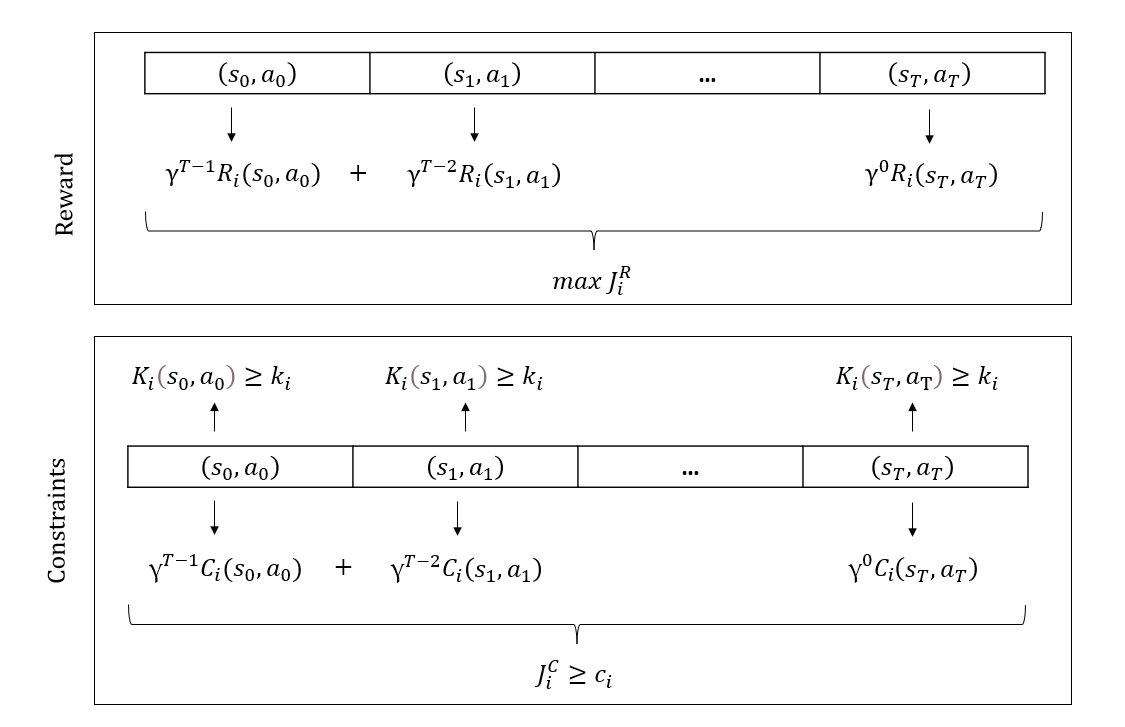}
    \caption{Pictorial representation of the problem to be solved.}
    
    \label{fig:problemFormulation}
\end{figure}

\subsection{Background}
\label{sec:background}
Before moving on to our proposed methodologies to solve the problem formulated in Section~\ref{sec:problemFormulation}, we discuss some preliminary concepts that are required to properly understand our research in this section.

\subsubsection{Policy Gradient and Variance Reduction}
\label{sec:policy_grad}
Policy gradient methods are a class of reinforcement learning algorithms that directly optimize the agent's policy to maximize the expected cumulative rewards.
The fundamental concept behind these methods is to adjust the parameters of the policy towards enhancing the expected return.
This is different from value-based methods such as Q-learning \cite{watkins1992q} and SARSA \cite{rummery1994line} that focus on learning the optimal value function, which estimates the expected cumulative reward starting from a given state under a fixed policy.
Policy gradient methods aim to directly optimize the policy itself, often bypassing the need to explicitly estimate the value function. This difference in approach makes policy gradient methods, such as the REINFORCE method, well suited for problems with continuous and high-dimensional action spaces.

Now, a key component of a policy gradient method is the gradient estimation method. For a parametrized policy $\pi_{\theta}(\tau)$, where $\tau$ represents a trajectory, the gradient of the expected reward $\mathbb{E}_{\theta}[R(\tau)]$ according to the REINFORCE method is given by Equation~\ref{eq:reinforce}.

\begin{equation}
    \nabla_{\theta} \mathbb{E}_{\theta}[R(\tau)] = \mathbb{E}[R(\tau)\, \nabla_{\theta} \log \pi_{\theta}(\tau)] \label{eq:reinforce}
\end{equation}

Here, $R(\tau)$ denotes the total reward obtained from trajectory $\tau$, and $\nabla_{\theta}$ represents the gradient with respect to the policy parameters $\theta$. The REINFORCE gradient is obtained by sampling trajectories and computing the gradient using these samples.

In the actor-critic method, an additional baseline term $b$ is introduced to reduce the variance of the gradient estimate. The baseline represents an estimate of the expected reward and helps stabilize the learning process. The gradient with respect to the policy parameters according to this method is given by Equation~\ref{eq:reinforce_ac}

\begin{equation}
    \nabla_{\theta} \mathbb{E}_{\theta}[R(\tau)] = \mathbb{E}[(R(\tau) - b) \nabla_{\theta} \log \pi_{\theta}(\tau)] \label{eq:reinforce_ac}
\end{equation}

In this formulation, the term $(R(\tau) - b)$ is subtracted from the total reward $R(\tau)$ to form the advantage. The advantage represents the difference between the observed reward and the baseline, providing a measure of how well the policy performed relative to the baseline.

\subsubsection{Primal Dual Formulation}
\label{sec:primal_dual_back}
Suppose that we want to maximize a reward function $J^{R}(\theta)$ parameterised by some variable $\theta$, subject to a constraint on the cost function $J^{C}(\theta)$ also parameterised by $\theta$ as Expression~\ref{eq:simple_prob} 
where $c$ is some constant.
\begin{align}
    \underset{\theta \in \Theta}{\text{maximize}}\ J^{R}(\theta) \quad \text{subject to} \quad J^{C}(\theta) \ge c \label{eq:simple_prob}
\end{align}
Then, using the method of Lagrange multipliers \cite{bertsekas2014constrained}, Expression~\ref{eq:simple_prob} can be reformulated as Expression~\ref{eq:simple_lagrange} using the Lagrangian $J^{L}$ 
where $\lambda$ denotes the dual variables.
\begin{equation}
    \underset{\theta \in \Theta}{\text{maximize}}\ 
    \underset{\lambda \ge 0}{\text{minimize}}\ 
    J^{L}(\theta, \lambda) := J^{R}(\theta) + \lambda(J^{C}(\theta) - c) \label{eq:simple_lagrange}
\end{equation} 
Thus, our original constrained maximization problem is now a min-max problem on the primal ($\theta$) and dual variables ($\lambda$). 

\subsubsection{Stochastic Gradient Descent and Ascent}
\label{sec:sgda_back}
One of the most prominent algorithms for solving min-max optimization problems of the form similar to Expression~\ref{eq:simple_lagrange} is Stochastic Gradient Descent Ascent (SGDA) \cite{beznosikov2023stochastic}. In its simplest form, small increments are added and scaled by the gradient of the function with respect to the variable over which we are trying to maximize the function, and small decrements are subtracted and scaled by the gradient with respect to the variable over which we are trying to minimize the function. Over several updates, the values of the variables converge toward the optimal values. So if this is applied for Expression~\ref{eq:simple_lagrange}, we get the following update rules at time $t$:
\begin{align}
    \theta^{(t+1)} = \theta^{(t)} + \eta_1 \nabla_\theta J^{\,L}(\theta, \lambda)\\
    \lambda^{(t+1)} = \lambda^{(t)} - \eta_2\nabla_\lambda J^{\,L}(\theta, \lambda) \label{eq:simple_update_lambda}.
\end{align}
Now, there is one further trick needed in order to ensure correctness considering the constricted domain of $\lambda \ge 0$. That is projected gradient descent. The new version of Equation \ref{eq:simple_update_lambda} would be Equation \ref{eq:simple_update_lambda_proj} where $\mathcal{P}_\Lambda$ denotes the projection operator defined for the domain $\Lambda$.
\begin{align}
    \lambda^{(t+1)} = \mathcal{P}_\Lambda(\lambda^{(t)} - \eta_2\nabla_\lambda J^{\,L}(\theta, \lambda)) \label{eq:simple_update_lambda_proj}
\end{align}
This operator essentially projects the value of its parameter onto the domain $\Lambda$. This ensures that the value of $\lambda$ always stays within its restricted domain $\Lambda$.

\subsubsection{Decentralized Optimization}
\label{sec:back_dec_op}
Suppose we have $N$ agents connected by a communication network $G$ and each agent $i$ is equipped with a function $F_i(\theta, \lambda)$, and we need to find the values of $\theta$ and $\lambda$ for which,
\begin{align}
    \underset{\{\theta \in \Theta\}}{\text{max}}\ 
    \underset{\{\lambda \in \Lambda\}}{\text{min}}\ 
    \frac{1}{n} \sum_{i=1}^{n} F(\theta, \lambda).
\end{align}

Now, a decentralized approach to solving this has each agent with its own version of $\theta$ and $\lambda$, namely $\theta_i$ and $\lambda_i$. Then the problem can be reformulated as in Equation~\ref{eq:dec_op_prob} where $\mathcal{N}_i$ denotes the set of neighbors of agent $i$.  It is imperative to enforce the condition $\theta_i = \theta_j$ to guarantee uniformity in the policy across all agents. This constraint ensures that the learning objectives remain consistent among all participating agents. 
\begin{align}
    \underset{\{\theta_i \in \Theta\}}{\text{max}}\ 
    \underset{\{\lambda_i \in \Lambda\}}{\text{min}}\ 
    \frac{1}{n} \sum_{i=1}^{n} F(\theta_i, \lambda_i) \label{eq:dec_op_prob}\\
    \text{s.t.}\quad \theta_i = \theta_j\quad j \in \mathcal{N}_i,\ \forall i \nonumber.
\end{align}

Then a stochastic gradient descent ascent approach similar to \cite{xian2021faster} can be applied. The basic idea is that each agent performs a local gradient update and then seeks consensus with its neighbors in order to fulfill the equality constraint. The overall process can be summarized by the steps: 

\begin{align}
    \theta_i^{t+1} := \sum_{j \in \mathcal{N}_i}\ W_{ij} (\theta_j^t + \eta_1 \nabla_{\theta_i}\ F(\theta_i^{t}, \lambda_i^{t}) \label{eq:back_theta_update_0},\\
    \lambda_i^{t+1} := \sum_{j \in \mathcal{N}_i}\ W_{ij} (\lambda_j^t - \eta_2 \nabla_{\lambda_i}\ F(\theta_i^{t}, \lambda_i^{t})) \label{eq:back_lambda_update_0}.
\end{align}

Here, Equation \ref{eq:back_theta_update_0} represents the gradient ascent applied to the primal variable, and Equation \ref{eq:back_lambda_update_0} refers to the descent applied to the dual variables. 
$\eta_1$ and $\eta_2$ are the corresponding learning rates. $W \in \mathbb{R}^{n \times n}$ is a weight matrix associated with the communication network $G$ and is doubly stochastic, that is, $W 1_n = 1_n$ and $1_n^TW = 1_n^T$.
The row-stochasticity of the weight matrix ensures consensus among all agents, while its column-stochasticity guarantees that the local gradient of each agent contributes equally to the global objective.

Despite its simplicity, vanilla decentralized gradient ascent/descent suffers from slow convergence. To remedy this, gradient tracking \cite{li2020communication, pu2021distributed} can be used to reduce variance and speed up convergence. The basic idea is to correct the biases between the local copies of the parameter by tracking the average gradient. The modified versions of Equations \ref{eq:back_theta_update_0}, \ref{eq:back_lambda_update_0} with gradient tracking are shown below, 

\begin{align}
    x_i^{t+1} := \sum_{j \in \mathcal{N}_i}\ W_{ij} (x_j^t + \nabla_{\theta_i}\ F(\theta_i^t, \lambda_i^t) - \nabla_{\theta_i}\ F(\theta_i^{t-1}, \lambda_i^{t-1})), \label{eq:back_theta_tracking}\\
    y_i^{t+1} := \sum_{j \in \mathcal{N}_i}\ W_{ij} (y_j^t + \nabla_{\lambda_i}\ F(\theta_i^t, \lambda_i^t) - \nabla_{\lambda_i}\ F(\theta_i^{t-1}, \lambda_i^{t-1})), \label{eq:back_lambda_tracking}\\
    \theta_i^{t+1} := \sum_{j \in \mathcal{N}_i}\ W_{ij} (\theta_j^t + \eta_1 x_j^{t+1}), \label{eq:back_theta_update}\\
    \lambda_i^{t+1} := \sum_{j \in \mathcal{N}_i}\ W_{ij} (\lambda_j^t - \eta_2 y_j^{t+1}). \label{eq:back_lambda_update}
\end{align}

Here, $x_i$, $y_i$ are the gradient tracking variables for $\theta_i$ and $\lambda_i$ respectively, 
Equation \ref{eq:back_theta_tracking} and Equation \ref{eq:back_lambda_tracking} represent the tracking steps, while Equation \ref{eq:back_theta_update} and Equation \ref{eq:back_lambda_update} refer to the gradient ascent and descent step respectively, similar to Equation \ref{eq:back_theta_update_0} and Equation \ref{eq:back_lambda_update_0}.

\section{The DePAint Algorithm}
\label{chap:proposedMethod}
The purpose of this section is to explain our proposed method for solving the problem outlined in Section~\ref{sec:problemFormulation}. We coin this novel algorithm DePAint (\textbf{De}centralized Policy Gradient considering \textbf{P}eak and \textbf{A}verage Constra\textbf{int}s). This section is primarily divided into two parts; the first part discussed in Section~\ref{sec:algo_prelim} outlines how we convert the problem stated in Section~\ref{chap:problemFormulationAndBackground} to an equivalent yet simpler form, and in Section \ref{sec:algorithm} we explain the steps in our algorithm which tries to tackle the aforementioned problem.

\subsection{Preliminaries}
\label{sec:algo_prelim}
In order to tackle the problem at hand, we transform it into an equivalent simpler problem. At first, in Section~\ref{ssec:policy_grad}, we introduce a way to approximate the agent's policy to handle large state and action spaces properly. Secondly, in Section~\ref{ssec:peak_deal}, we augment the long-term utility function so that it accounts for both the average and peak constraint, leaving us with only one constraint to deal with. Finally, in Section~\ref{ssec:primal_dual}, using the method of Lagrange multipliers, we convert the constrained maximization problem into a min-max optimization problem.

\subsubsection{Deep Policy Gradient}
\label{ssec:policy_grad}
The space of all possible policy functions can be rather large; as a result, it is intractable to exactly define the policy function. As such we use parameterized function approximators such as deep neural networks to represent each agent's policy. Using a deep neural network as a policy function approximator reduces the problem of finding the optimal policy to finding the optimal parameters for the policy function approximator i.e. the weights of the neural network. So, $\pi_i: S \rightarrow \Delta(A)$ is parameterized as $\pi_{\theta_{[i]}}$ by a parameter $\theta_{[i]} \in \R^{d_i}$ with dimension $d_i$. So, we can define the joint policy as shown in Equation~\ref{eq:param_policy} where, $\theta = [\theta_{[1]}^T\ ...\ \theta_{[n]}^T]^T \in \R^d$ and $d = \sum_{i = 1}^n d_i$.
\begin{equation}
    \pi_{\theta}(a^t | s^t) := \prod_{i = 1}^n \pi_{\theta_{[i]}}(a_i^t | s^t) \label{eq:param_policy}
\end{equation}
As the agents are required to train independently, each agent $i$ keeps a local version of the parameter $\theta$, named $\theta_i$.
Based on this we can modify Equations \ref{eq:traj_prob}, \ref{eq:disc_reward} and \ref{eq:disc_util} as,
\begin{align}
    p(\tau | \theta) := \rho(s^0) \prod_{t=0}^\infty \pi_\theta(a^t | s^t) P(s^{t+1}|s^t, a^t) \label{eq:traj_prob_theta}\\
    J_i^{R}(\theta) := \underset{\tau \sim p(.|\theta)}{\E} \left[\sum_{t = 0}^{\infty} \gamma^t R_i (s^t, a^t) \right]
    \label{eq:disc_reward_theta}\\
    J_i^{C}(\theta) := \underset{\tau \sim p(.|\theta)}{\E} \left[\sum_{t = 0}^{\infty} \gamma^t C_i (s^t, a^t) \right].
    \label{eq:disc_util_theta}
\end{align}
That leads to the problem of finding the optimal policy $\pi$ defined in Section~\ref{sec:problemFormulation} becoming a problem of finding the optimal parameters $\theta$ for which,

\begin{align}
    & \underset{\theta \in \R^d}{\max}\ \frac{1}{n} \sum_{i=0}^n J_i^{R}(\theta) \label{eq:prob_theta}\\
    \text{subject to }\ & J_i^{C}(\theta) \ge c_i & \forall i \in \N \label{eq:avg_const_theta}\\
    \text{and }\ & K_i(s^t, a^t) \ge k_i & \forall i \in \N, (s^t, a^t) \in p(.|\theta). \label{eq:peak_const_theta}
\end{align}

\subsubsection{Dealing with Peak Constraint}
\label{ssec:peak_deal}
Dealing with two types of constraints simultaneously is difficult because peak and average constraints have pretty diverse forms.  We solve this problem by augmenting the average constraint to include the peak constraint.
We define the new long-term utility function as Equation~\ref{eq:new_peak} where $C_i(s, a)$ is the long-term utility function we defined in Section~\ref{sec:problemFormulation}, $M$ is a large number (greater than $max\ J_i^C - c_i$), and $I_i$ is an identity function defined in Equation \ref{eq:new_peak_identity}.

\begin{equation}
    \widehat{C}_i(s, a) := C_i(s, a) - M I_i(s, a) \label{eq:new_peak}.
\end{equation}
\begin{equation}
    I_i(s, a) := 
    \begin{cases}
    \mathbf{1}_{m_1},& \text{if } K_i(s, a) < k_i\\
    0,              & \text{otherwise} \label{eq:new_peak_identity}
\end{cases}
\end{equation}

For a finite horizon, $max\ J_i^C - c_i$ is bounded. As such, we can always define a suitable $M$ for a given environment. This formulation ensures the fact that when the peak constraint is violated i.e. $K_i(s, a) < k_i$, the new average constraint is also violated, $\widehat{C}_i(s, a) < c_i$. As a result, we no longer need to handle the peak constraint explicitly. We can also define a new $J_i^{\widehat{C}}(\pi)$ based on the newly defined $\widehat{C}$.
\begin{equation}
     J_i^{\widehat{C}}(\pi) := \underset{\tau \sim p(.|\pi)}{\E} \left[\sum_{t = 0}^{\infty} \gamma^t \widehat{C}_i (s^t, a^t) \right]
     \label{eq:new_lt_util}
\end{equation}
From this point onward, for the sake of simplicity, we are going to use $C$ and $J_i^{C}(\pi)$ to refer to $\widehat{C}$ and $J_i^{\widehat{C}}(\pi)$ respectively.

\subsubsection{Primal Dual Formulation}
\label{ssec:primal_dual}
Now that we only have to deal with the average constraint, the original problem can be written as,
\begin{align}
    \underset{\theta \in \Theta}{\text{max}}\ \frac{1}{n} \sum_{i=1}^{n} J_i^{R}(\theta) \ \  \text{subject to} \ \ J_i^{C}(\theta) \ge c_i \ \ \forall i \in \mathcal{N}. \label{eq:no_peak_cmdp}
\end{align}
Maximizing the objective while dealing with the constraints is not a trivial task. Therefore as discussed in Section~\ref{sec:primal_dual_back} using the method of Lagrange multipliers, we can formulate this constrained maximization problem as defined in Expression \ref{eq:no_peak_cmdp} into an equivalent unconstrained min-max problem for the associated Lagrangian $J^{L}(\theta, \lambda_1, ..., \lambda_n)$ in Expression~\ref{eq:multi_primal_dual} where $\theta$ is the primal variable and $\lambda_1, ..., \lambda_n$ are the non-negative Lagrange multipliers or dual variables.
\begin{align}
    \underset{\theta \in \Theta}{\text{max}}\ 
    \underset{\{\lambda_i \ge 0\}}{\text{min}}\ 
    J^{L}(\theta, \lambda_1, ..., \lambda_n) := \frac{1}{n} \sum_{i=1}^{n} J_i^{R}(\theta) + \sum_{i=1}^{n} \lambda_i(J_i^{C}(\theta) - c_i) \label{eq:multi_primal_dual}
\end{align}
Now the problem is reduced to maximizing on the primal variable $\theta$ and minimizing on the dual variables $\{\lambda_i\}_{i=1}^n$. That enables us to use gradient descent ascent based methods to solve the problem.

\subsection{The Algorithm Description}
\label{sec:algorithm}
Now that we have transformed the problem into a more suitable form, we explain our proposed algorithm for solving this problem. The steps of the algorithm are outlined in the pseudocode defined in Algorithm~\ref{alg:depaint}.

\SetKwFor{ForParallel}{for}{do in parallel}{endfor}

\begin{algorithm*}
\caption{DePAint} \label{alg:depaint}
    \SetKwInOut{Input}{Input}
    \SetKwInOut{Output}{Output}
    \Input{number of iterations $T$, horizon $H$, batch size $B$, primal learning rate $\eta_1$, dual learning rate $\eta_2$, primal momentum parameter $\beta_1$, dual momentum parameter $\beta_2$, weight matrix $W$}
    \Output{optimal policy parameter $\theta_i^{T+1}$}
    \BlankLine

    \SetKwFunction{UpdateCritics}{UpdateCritics}
    \SetKwFunction{EstimateGradients}{EstimateGradients}
    \SetKwFunction{Momentum}{Momentum}
    \SetKwFunction{UpdateTracking}{UpdateTracking}
    \SetKwFunction{UpdateParams}{UpdateParams}

    \ForParallel {$i \in \mathcal{N}$}{
        
        Sample $B$ finite-horizon trajectories $\{\tau_i^b\}_{b=1}^B$ of size $H$ from $p(.|\theta_i^0)$\; \label{alg_line:gen_traj_init}
        Update critic networks using $\{\tau_i^b\}_{b=1}^B$\; \label{alg_line:update_critic_init}
        \BlankLine
        
        $u_i^0, v_i^0 \gets$ \EstimateGradients{$\{\tau_i^b\}_{b=1}^B$}\; \label{alg_line:estimate_grad_init}
        \BlankLine
        
        $x_i^1 \gets$ \UpdateTracking{$W, \{0, u_j^0, 0\}_{j \in \mathcal{N}_i}$}\; \label{alg_line:grad_track_init_1}
        $y_i^1 \gets$ \UpdateTracking{$W, \{0, v_j^0, 0\}_{j \in \mathcal{N}_i}$}\; \label{alg_line:grad_track_init_2}
        \BlankLine
        
        \For {$t\leftarrow 1$ \KwTo $T$} {
            Sample $B$ finite-horizon trajectories $\{\tau_i^b\}_{b=1}^B$ of size $H$ from $p(.|\theta_i^0)$\; \label{alg_line:gen_traj}
            Update critic networks using $\{\tau_i^b\}_{b=1}^B$\; \label{alg_line:update_critic}
            \BlankLine
            
            $u_i^t, v_i^t \gets$ \EstimateGradients{$\{\tau_i^b\}_{b=1}^B, \beta_1, \beta_2, u_i^{t-1}, v_i^{t-1}$}\; \label{alg_line:estimate_grad}
            \BlankLine
            
            $x_i^{t+1} \gets$ \UpdateTracking{$W, \{x_j^t, u_j^t, u_j^{t-1}\}_{j \in \mathcal{N}_i}$}\;
            \label{alg_line:grad_track_1}
            $y_i^{t+1} \gets$ \UpdateTracking{$W, \{y_j^t, v_j^t, v_j^{t-1}\}_{j \in \mathcal{N}_i}$}\;
            \label{alg_line:grad_track_2}
            \BlankLine
            
            $\theta_i^{t+1}, \lambda_i^{t+1} \gets $ \UpdateParams{$\eta_1, \eta_2, W, \{\theta_j^t, \lambda_j^t, x_j^{t+1}, y_j^{t+1}\}_{j \in \mathcal{N}_i}$}\; \label{alg_line:update_param_2}
        }
    }
    \KwRet{$\{\theta_i^{T+1}\}_{i \in \mathcal{N}}$}
\end{algorithm*}

The algorithm primarily consists of the initialization part and the main loop. The main loop consists of four parts.  The first part is the exploration defined in line \ref{alg_line:gen_traj}, which consists of the agent generating some finite trajectories from the given environment. Then, in the second step in line \ref{alg_line:update_critic}, we use these trajectories to update the critic network. Thirdly, in line \ref{alg_line:estimate_grad}, we generate gradient estimators for the Lagrangian function with respect to our $\theta$ and $\lambda$. And finally in lines \ref{alg_line:grad_track_1}, \ref{alg_line:grad_track_2} and \ref{alg_line:update_param_2}, we update the parameters ($\theta$ and $\lambda$) and communicate with neighbouring agents to establish consensus. The initialization process is very similar to the steps inside the main loop, where we generate some trajectories and estimate the initial gradients to initialize the tracking variables. It is noteworthy to mention that, each agent executes in parallel but needs to synchronize the gradient with other agents at each iteration of the main loop, specifically on lines \ref{alg_line:grad_track_1}, \ref{alg_line:grad_track_2} and \ref{alg_line:update_param_2}. 

In the following sections, we explain the gradient estimation process and the parameter update step in more detail.

\subsubsection{Gradient Estimation}
In light of Equation \ref{eq:multi_primal_dual}, the required gradients $\nabla_{\theta_i}\ J_i^L(\theta_i^t, \lambda_i^t)$ and $\nabla_{\lambda_i}\ J_i^L(\theta_i^t, \lambda_i^t)$ can be derived as follows,
\begin{align}
    \nabla_{\theta_i}\ J_i^L(\theta_i, \lambda_i) = \nabla_{\theta_i} J_i^{R}(\theta_i) + \lambda \nabla_{\theta_i} J_i^{C}(\theta_i) \label{eq:lagrange_grad_theta}\\
    \nabla_{\lambda_i}\ J_i^L(\theta_i, \lambda_i) = J_i^{C}(\theta_i) - c_i \label{eq:lagrange_grad_lambda}
\end{align}
For estimating $\nabla_{\theta_i}\ J_i^L(\theta_i, \lambda_i)$ we need to first estimate $\nabla_{\theta_i} J_i^{R}(\theta_i)$ and $\nabla_{\theta_i} J_i^{C}(\theta_i)$.

\paragraph{REINFORCE}
Without having full knowledge about the dynamics of the environment beforehand, we have no way to calculate the exact value of $\nabla_{\theta_i} J_i^{R}(\theta_i)$ and $\nabla_{\theta_i} J_i^{C}(\theta_i)$.  To estimate $\nabla_{\theta_i} J_i^{R}(\theta_i)$ and $\nabla_{\theta_i} J_i^{C}(\theta_i)$, we adopt the REINFORCE policy gradient estimator with baselines as discussed in Section \ref{sec:policy_grad}. Let's use the baselines $\hat{b}_i^R$ and $\hat{b}_i^C$ for $\nabla_{\theta_i} J_i^{R}(\theta_i)$ and $\nabla_{\theta_i} J_i^{C}(\theta_i)$ respectively. From equation \ref{eq:reinforce_ac} we can derive Equation \ref{eq:obj_grad_est_inf} and \ref{eq:util_grad_est_inf} ($\widehat{\nabla}$ is used to express a gradient estimate) where $\tau := (s_0, a_0, s_1, ..., )$ represents an infinite trajectory produced under the policy $\pi_{\theta_i}$. 
\begin{align}
    \widehat{\nabla}_{\theta} J_i^{R}(\theta_i) := \left[ \sum^{\infty}_{h=0} \nabla_{\theta_i} \log \pi_{\theta_i}(a^h|s^h)\right] \cdot \left[  \sum^{\infty}_{h=0}\gamma^h R_i(a^h, s^h) - \hat{b}_i^R\right] \label{eq:obj_grad_est_inf}\\
    \widehat{\nabla}_{\theta} J_i^{C}(\theta_i) := \left[ \sum^{\infty}_{h=0} \nabla_{\theta_i} \log \pi_{\theta_i}(a^h|s^h)\right] \cdot \left[  \sum^{\infty}_{h=0}\gamma^h C_i(a^h, s^h) - \hat{b}_i^C\right] \label{eq:util_grad_est_inf}
\end{align}
For the baseline, the state-value function is used. The state-value function for $R_i$ and $C_i$ may be defined as in Equation~\ref{eq:state_value_r} and \ref{eq:state_value_c}. $V_{i}^{R}(s)$ and $V_{i}^{C}(s)$ are usually approximated using separate neural networks, which are trained simultaneously with the policy network in lines \ref{alg_line:update_critic_init} and \ref{alg_line:update_critic}. Let us term these as \textit{critic networks} and denote them as $\widehat{V}_{i}^{R}(s)$ and $\widehat{V}_{i}^{C}(s)$.

\begin{align}
    V_{i}^{R}(s) := \underset{s^0 = s,\ \tau \sim p(.|\theta_i)}{\mathbb{E}} \left[\sum_{t=0}^\infty \gamma^t R_i(s^t, a^t) \right] \label{eq:state_value_r}\\
    V_{i}^{C}(s) := \underset{s^0 = s,\ \tau \sim p(.|\theta_i)}{\mathbb{E}} \left[\sum_{t=0}^\infty \gamma^t C_i(s^t, a^t) \right] \label{eq:state_value_c}
\end{align}

Now, although Equations \ref{eq:obj_grad_est_inf} and \ref{eq:util_grad_est_inf} give us a reasonable estimate of the required gradients, we cannot calculate them in practice due to requiring an infinite trajectory. Therefore, we use finite horizon approximation with horizon $H$. Moreover, to further reduce variance, we can use a minibatch of $B$ trajectories $\{\tau_i^b\}_{b=1}^B$ instead of a single trajectory and take the average gradient over them. As a result, with mini-batch and finite-horizon estimation, Equations \ref{eq:obj_grad_est_inf} and \ref{eq:util_grad_est_inf} become

\begin{align}
    \widehat{\nabla}_{\theta} J_i^{R}(\theta_i) := \frac{1}{B} \sum_{b=1}^B\ \left[ \sum^{H-1}_{h=0} \nabla_{\theta_i} \log \pi_{\theta_i}(a^{h, b}|s^{h, b})\right] \cdot
    \left[  \sum^{H-1}_{h=0}\gamma^h R_i(a^{h, b}, s^{h, b}) - \widehat{V}_{i}^{R}(s^{h, b})\right] \label{eq:obj_grad_est}\\
    \widehat{\nabla}_{\theta} J_i^{C}(\theta_i) := \frac{1}{B} \sum_{b=1}^B\ \left[ \sum^{H-1}_{h=0} \nabla_{\theta_i} \log \pi_{\theta_i}(a^{h, b}|s^{h, b})\right] \cdot
    \left[  \sum^{H-1}_{h=0}\gamma^h C_i(a^{h, b}, s^{h, b}) - \widehat{V}_{i}^{C}(s^{h, b})\right] \label{eq:util_grad_est}
\end{align}

Therefore, according to Equation \ref{eq:lagrange_grad_theta} we can calculate the gradient estimate $\widehat{\nabla}_{\theta_i}\ J_i^L(\theta_i, \lambda_i)$ as,

\begin{align}
    \widehat{\nabla}_{\theta_i}\ J_i^L(\theta_i, \lambda_i) = \widehat{\nabla}_{\theta_i} J_i^{R}(\theta_i) + \lambda \widehat{\nabla}_{\theta_i} J_i^{C}(\theta_i). \label{eq:lagrange_grad_theta_est}
\end{align}

On the other hand, we can also calculate an estimate of $\nabla_{\lambda_i}\ J_i^L(\theta_i^t, \lambda_i^t)$ by calculating an approximation of $J_i^{C}(\theta_i)$ using the minibatch and finite-horizon techniques similar to Equation \ref{eq:obj_grad_est} and \ref{eq:util_grad_est}. Given a batch of $B$ finite-horizon trajectories of size $H$ $\{\tau_i^b\}_{b=1}^B$, we can define $\nabla_{\lambda_i}\ J_i^L(\theta_i^t, \lambda_i^t)$ as in Equation~\ref{eq:lagrange_grad_lambda_est}, where $\widehat{J_i^{C}}(\theta_i)$ is used to indicate an estimate of $J_i^{C}(\theta_i)$.
\begin{align}
    \widehat{J_i^{C}}(\theta_i) := \frac{1}{B}\sum_{b=1}^{B}\sum_{h = 0}^{H-1} \gamma^h C_i (s^{h, b}, a^{h, b})\\
    \widehat{\nabla}_{\lambda_i}\ J_i^L(\theta_i, \lambda_i) := \widehat{J_i^{C}}(\theta_i) - c_i \label{eq:lagrange_grad_lambda_est}
\end{align}

\paragraph{Momentum Based Variance Reduction}
If the estimates from Equations \ref{eq:lagrange_grad_theta_est} and \ref{eq:lagrange_grad_lambda_est} are directly applied in the gradient descent process (discussed in the next section) the high variance results in slow convergence. To further reduce the variance, we apply the momentum-based variance reduction technique \cite{cutkosky2019momentum, tran2019hybrid}. The momentum technique ensures the stability of the gradient estimates by smoothing its values across multiple learning steps. Let us use $u_i^t$ and $v_i^t$ to denote the REINFORCE based gradient estimates of $\nabla_{\theta_i}\ J_i^L(\theta_i, \lambda_i)$ and $\nabla_{\lambda_i}\ J_i^L(\theta_i, \lambda_i)$ respectively at time $t$. We use $\widehat{u}_i^t$ and $\widehat{v}_i^t$ to represent their momentum-based counterparts. Now, the modified gradient estimator according to the aforementioned technique can be written as:
\begin{align}
    \widehat{u}_i^t &:= \beta_1\, u_i^t + (1-\beta_1) (\widehat{u}_i^{t-1} + u_i^t - u_i^{t-1})\\
    \widehat{v}_i^t &:= \beta_2\, v_i^t + (1-\beta_2) (\widehat{v}_i^{t-1} + v_i^t - v_i^{t-1})
\end{align}
Here $\beta_1,\beta_2 \in (0, 1]$ represents the momentum parameters. For this method to work properly, the gradient estimators need to be unbiased. However, in Equations \ref{eq:obj_grad_est}, \ref{eq:util_grad_est}, the distribution $p(.|\theta_i^t)$ controls the sampled trajectory $\tau_i^t$. It is clear that the estimator $g_i(\tau_i^t|\theta_i^{t-1})$ is biased with respect to $\nabla V(\theta_i^{t-1})$. We can use the importance sampling technique to ensure the unbiased property
\begin{equation}
    \mathbb{E}_{\tau_i^t\mathtt{\sim} p(\cdot|\theta_i^t)}\{\omega(\tau_i^t|\theta_i^{t-1},\theta_i^{t})g_i(\tau_i^t|\theta_i^{t-1}) \} = \nabla V_i(\theta_i^{t-1})
\end{equation}

Here, importance weight is denoted as $\omega(\tau_i^t|\theta_i^{t-1},\theta_i^{t-1})$. It's defined as:
\begin{equation}
    \omega(\tau_i^t|\theta_i^{t-1},\theta_i^{t}) := \frac{p(\tau_i^t|\theta_i^{t-1})}{p(\tau_i^t|\theta_i^t)} = \prod_{h=1}^{H-1} \frac{\pi_{\theta_i^{t-1}}(a^h|s^h)}{\theta_i^{t}(a^h|s^h)}
\end{equation}

So, the momentum-based variance reduction for the policy gradient of the i-th agent is given by:
\begin{align}
    \widehat{u}_i^t &:= \beta_1\,u_i^t + (1-\beta_1) (\widehat{u}_i^{t-1} + u_i^t - \omega(\theta_i^{t-1},\theta_i^{t})\cdot u_i^{t-1} ) \label{eq:momentum_obj_grad}\\
    \widehat{v}_i^t &:= \beta_2\,v_i^t + (1-\beta_2)(\widehat{v}_i^{t-1} + v_i^t - \omega(\theta_i^{t-1},\theta_i^{t})\cdot v_i^{t-1}) \label{eq:momentum_util_grad}
\end{align}

In light of the above discussion, the function $EstimateGradients$ used in lines \ref{alg_line:estimate_grad_init} and \ref{alg_line:estimate_grad} can be defined as in Procedure~\ref{proc:estimate_gradients}.\\

\begin{procedure} [H]
\caption{EstimateGradients()} \label{proc:estimate_gradients}
    \KwData{$B$ finite-horizon trajectories of size $H$ $\{\tau_i^b\}^B_{b=1}$}
    \KwResult{Gradient estimates of $\nabla_{\theta_i}\ J_i^L(\theta_i^t, \lambda_i^t)$ and $\nabla_{\lambda_i}\ J_i^L(\theta_i^t, \lambda_i^t)$}
    Calculate REINFORCE gradient estimates $\widehat{\nabla}_{\theta} J_i^{R}(\theta_i)$ and $\widehat{\nabla}_{\theta} J_i^{C}(\theta_i)$ from $\{\tau_i^b\}^B_{b=1}$ using Equations \ref{eq:obj_grad_est} and \ref{eq:util_grad_est} respectively\;
    Calculate $\widehat{\nabla}_{\theta_i}\ J_i^L(\theta_i, \lambda_i)$ and $\widehat{\nabla}_{\lambda_i}\ J_i^L(\theta_i, \lambda_i)$ using Equations \ref{eq:lagrange_grad_theta_est} and \ref{eq:lagrange_grad_lambda_est}, and denote them as $u_i^t$ and $v_i^t$ respectively\;
    Calculate momentum based estimates $\widehat{u}_i^t$ and $\widehat{v}_i^t$ using Equations \ref{eq:momentum_obj_grad} and \ref{eq:momentum_util_grad} respectively\;
    \KwRet{$\widehat{u}_i^t, \widehat{v}_i^t$}
\end{procedure}

\subsubsection{Decentralized Optimization}
The Expression \ref{eq:multi_primal_dual} would be straightforward to solve if it were in a centralized setting. However, our setting is decentralized, where the agents are connected by a communication network. 

We can reformulate the problem as a decentralized min-max optimization problem of the form \ref{eq:multi_primal_dual_short} where $J_i^{L}(\theta_i, \lambda_i) = J_i^{R}(\theta_i) + \lambda_i(J_i^{C}(\theta_i) - c)$ and $\mathcal{N}_i$ denotes the set of neighbors of agent $i$.
\begin{align}
    \underset{\{\theta_i \in \Theta\}}{\text{max}}\ 
    \underset{\{\lambda_i \ge 0\}}{\text{min}}\ 
    \frac{1}{n} \sum_{i=1}^{n} J_i^{L}(\theta_i, \lambda_i) \label{eq:multi_primal_dual_short}\\
    \text{s.t.}\quad \theta_i = \theta_j\quad j \in \mathcal{N}_i,\ \forall i \nonumber
\end{align}
This is similar to the problem discussed in Section~\ref{sec:back_dec_op}. So we can use stochastic gradient descent ascent and gradient tracking as discussed there, resulting in the update steps,

\begin{align}
    x_i^{t+1} := \sum_{j \in \mathcal{N}_i}\ W_{ij} (x_j^t + \nabla_{\theta_i}\ J_i^L(\theta_i^t, \lambda_i^t) - \nabla_{\theta_i}\ J_i^L(\theta_i^{t-1}, \lambda_i^{t-1})), \label{eq:theta_tracking}\\
    y_i^{t+1} := \sum_{j \in \mathcal{N}_i}\ W_{ij} (y_j^t + \nabla_{\lambda_i}\ J_i^L(\theta_i^t, \lambda_i^t) - \nabla_{\lambda_i}\ J_i^L(\theta_i^{t-1}, \lambda_i^{t-1})), \label{eq:lambda_tracking}\\
    \theta_i^{t+1} := \sum_{j \in \mathcal{N}_i}\ W_{ij} (\theta_j^t + \eta_1 x_j^{t+1}), \label{eq:theta_update}\\
    \lambda_i^{t+\frac{1}{2}} := \sum_{j \in \mathcal{N}_i}\ W_{ij} (\lambda_j^t - \eta_2 y_j^{t+1}),\quad
    \lambda_i^{t+1} := \mathcal{P}_\Lambda(\lambda_i^{t+\frac{1}{2}}). \label{eq:lambda_update}
\end{align}

Here, $x_i$, $y_i$ are the gradient tracking variables for $\theta_i$ and $\lambda_i$, respectively. Note that Equation \ref{eq:lambda_update} consists of a projection step to ensure that the updated $\lambda$ stays in the non-negative $\Lambda$ space.

Lines \ref{alg_line:grad_track_1} and \ref{alg_line:grad_track_2} in the main loop, as well as lines \ref{alg_line:grad_track_init_1} and \ref{alg_line:grad_track_init_2} in the initialization part constitute the gradient tracking part. While line \ref{alg_line:update_param_2} refers to the actual parameter update step. The procedures $UpdateTracking$ and $UpdateParams$ may be formally defined as in Procedure~\ref{proc:update_tracking}.\\

\begin{procedure} [H]
\caption{UpdateTracking()} \label{proc:update_tracking}
    \KwData{weight matrix $W$, \{current tracking variable $x_j^t$, current gradient estimate $u_j^t$, previous gradient estimate $u_j^{t-1}$\} $\forall j \in \mathcal{N'}$ (where $\mathcal{N'} \subset \mathcal{N}$)}
    \KwResult{next tracking variable for current agent $x_i^{t+1}$}

    $x_i^{t+1} \gets \sum_{j \in \mathcal{N'}}\ W_{ij} (x_j^t + u_j^t - u_j^{t-1})$\;

    \KwRet{$x_i^{t+1}$}\;
    
\end{procedure}

\hfill \break

\begin{procedure} [H]
\caption{UpdateParams()} \label{proc:update_params}
    \KwData{primal learning parameter $\eta_1$, dual learning parameter $\eta_2$, weight matrix $W$, \{ current parameter primal parameter $\theta_j^t$, current dual parameter $\lambda_j^t$, next primal tracking variable $x_j^{t+1}$, next dual tracking variable $y_j^{t+1}$ \} $\forall j \in \mathcal{N'}$ (where $\mathcal{N'} \subset \mathcal{N}$)}
    \KwResult{next primal parameter for current agent $\theta_i^{t+1}$, next dual parameter for current agent $\lambda_i^{t+1}$}

    $\theta_i^{t+1} \gets \sum_{j \in \mathcal{N}_i}\ W_{ij} (\theta_j^t + \eta_1 x_j^{t+1})$\;
    $\lambda_i^{t+\frac{1}{2}} \gets \sum_{j \in \mathcal{N'}}\ W_{ij} (\lambda_j^t - \eta_2 y_j^{t+1})$\;
    $\lambda_i^{t+1} \gets \mathcal{P}_\Lambda(\lambda_i^{t+\frac{1}{2}})$\;

    \KwRet{$\theta_i^{t+1}, \lambda_i^{t+1}$}
\end{procedure}

\section{Theoretical Analysis}
\label{chap:theoreticalProof}
In this section, we present a theoretical analysis of DePAint. If we recap our algorithm, we convert the problem into an unconstrained min-max problem over $J^L(\theta, \lambda)$ and apply a momentum-based stochastic gradient ascent approach with gradient tracking. Before implying anything about the convergence of this approach, we give a set of standard assumptions.\\

To begin with, we assume that the long-term discounted reward, utility, and peak violation functions satisfy a Lipschitz continuous condition (before augmentation to incorporate peak constraints), where the long-term peak violation function can be defined as, 
\begin{align}
J_i^I(\theta_i) = \underset{\tau \sim p(.|\theta_i)}{\mathbb{E}} \left[ \sum_{t=0}^{\infty} \gamma^t I_i(s^t, a^t) \right]
\end{align}
\begin{assumption}
\label{ass:lip_cont}
    Functions $J_i^R(\theta_i)$, $J_i^C(\theta_i)$ and $J_i^I(\theta_i)$ have Liptschitz continuity with respect to $\theta_i, \forall i$.
\end{assumption}

Next, we assume that the graph underlying the communication network is well-connected so that the consensus step can be performed in a decentralized way.
\begin{assumption}
\label{ass:graph_con}
    The graph is strongly connected, i.e., $W$ is a double stochastic matrix such that $\lambda_{max}(W) \in [0, 1)$, where $\lambda_{max}(W)$ signifies the second largest eigenvalue of the weight matrix $W$.
\end{assumption}

Lastly, we assume that the variance incurred due to importance sampling is bounded.
\begin{assumption}
\label{ass:var_imp_samp}
    The variance of importance sampling weight $\omega(\tau|x_1,x_2)$ is bounded, i.e., there exists a constant $\mathcal{M} > 0$ such that $\text{Var}(\omega(\tau|x_1,x_2)) \leq \mathcal{M}$, for any $x_1,x_2 \in \mathbb{R}^{d_i}$ and $\tau \sim p(.|x_1)$.
\end{assumption}
\hfill

Now, we can infer from Assumption~\ref{ass:lip_cont} that $J_i^L(\theta_i, \lambda_i)$ before augmenting $C$, will also hold the L-Lipschitz continuity property with respect to $\theta_i$ and $\lambda_i$ as it is merely a linear combination of $J_i^C(\theta_i)$ and $J_i^R(\theta_i)$. However, we claim that this property holds after augmenting $C$ as well.

\begin{lemma}
\label{lemma:after_peak_cont}
    If Assumption~\ref{ass:lip_cont} holds, $J_i^{\widehat{C}}(\theta_i, \lambda_i)$ is also Lipschitz continuous with respect to $\theta_i, \lambda_i, \forall i$
\end{lemma}
\begin{proof}
In order to prove Lemma~\ref{lemma:after_peak_cont}, it is enough to prove that $J_i^{\widehat{C}}(\theta_i)$ is Liptchitz continuous with respect to $\theta_i$. Given a bounded horizon of size $H$, we get Equation \ref{eq:lip_cont_proof} for all $x, y \in \mathbb{R}^{d_i}$, where $L_C$ and $L_I$ are the Lipchitz constants for $J_i^C$ and $J_i^I$ respectively.
\begin{align}
\label{eq:lip_cont_proof}
    ||J_i^{\widehat{C}}(x) - J_i^{\widehat{C}}(y)|| &\leq ||J_i^C(x) - J_i^C(y)|| + M H ||J_i^I(x) - J_i^I(y)|| \\
     &\leq L_C||x-y|| + M H L_I||x-y|| \nonumber\\
    &\leq (L_C + M H L_I)||x-y|| \nonumber
\end{align}
So, we can concur that $J_i^{\widehat{C}}(\theta_i)$ is L-Liptchitz continuous with respect to $\theta_i$, where $L=(L_C + M H L_I)$. As a result, $J_i^{\widehat{C}}(\theta_i, \lambda_i)$ is also Lipschitz continuous with respect to $\theta_i, \lambda_i$.
\end{proof}

\begin{theorem}[Sampling Complexity]
\label{theo:convergence}
    If Assumptions \ref{ass:lip_cont}, \ref{ass:graph_con} and \ref{ass:var_imp_samp} hold, the algorithm converges and the sampling complexity to reach an $\epsilon$-stationary point is $\mathcal{O}(N^{-1}\epsilon^{-3})$.
\end{theorem}
\begin{proof}
    Now, considering we have a Lipchitz continuous objective function, we can prove that a decentralized momentum-based stochastic gradient tracking algorithm converges. More specifically according to Jiang et al. (2022) \cite{jiang2022mdpgt}, for a bounded step size, it has a sampling complexity of $\mathcal{O}(N^{-1}\epsilon^{-3})$ to reach an $\epsilon$-stationary point (i.e. $\mathbb{E}[||J_i^L(x)||] < \epsilon$), where $N$ is the total number of agents.
\end{proof}
\begin{corollary}[Network Topology Invariance]
    From Theorem~\ref{theo:convergence}, we can infer that the sampling complexity is independent of the network topology.
\end{corollary}
\begin{proof}
    Moreover, as the sampling complexity is not dependent on $W$, we can also state that the convergence of the algorithm is independent of the network topology.
\end{proof}

\section{Empirical Evaluation}
\label{sec:empiricalEval}
In this section, we present the empirical evaluation of our algorithm, encompassing experiment settings to outline the test environment, implementation details, model parameters to elucidate our algorithm's configuration, and results to provide a comprehensive overview of its performance.

\subsection{Experiment Settings}
\subsubsection{Environment}
\label{sec:environment}
 For our empirical analysis, we chose two specific environments.  The first one is a modified version of the cooperative navigation environment for our empirical analysis. As a benchmark multi-agent environment, the cooperative navigation environment has been modified in several previous works, such as \cite{zhang2018fully,jiang2022mdpgt} to be compatible with the collaborative RL setting. Here, there are $n$ agent and landmark pairs. Each agent is required to find its corresponding landmark. The agents can take five actions at each step: up, down, left, right, or stay. The received reward of each agent is $-(\text{distance to landmark}) - \sum \textbf{1}_{\text{\{if colliding with an agent\}}}$. The long-term utility received by each agent is the distance of the closest agent. Meanwhile, the peak constraint is a constraint on its position. It is usually a bounding box that contains all the landmarks and agents.  A pictorial representation of the environment is shown in Figure~\ref{fig:coopnav_env}.

Another environment is a modified version of the predator-prey environment. Here, there are two types of agents, predators and preys. The predators gain a reward of $+1$ for colliding with a prey, while the prey gets a $-1$ reward (penalty) when it collides with a predator. The long-term utility received by each agent is the distance of the closest agent of the same type. The peak constraint is similar to the cooperative navigation environment. This environment is illustrated in Figure~\ref{fig:predprey_env}.

\begin{figure}[h]
    \centering
    \includegraphics[width=0.65\textwidth]{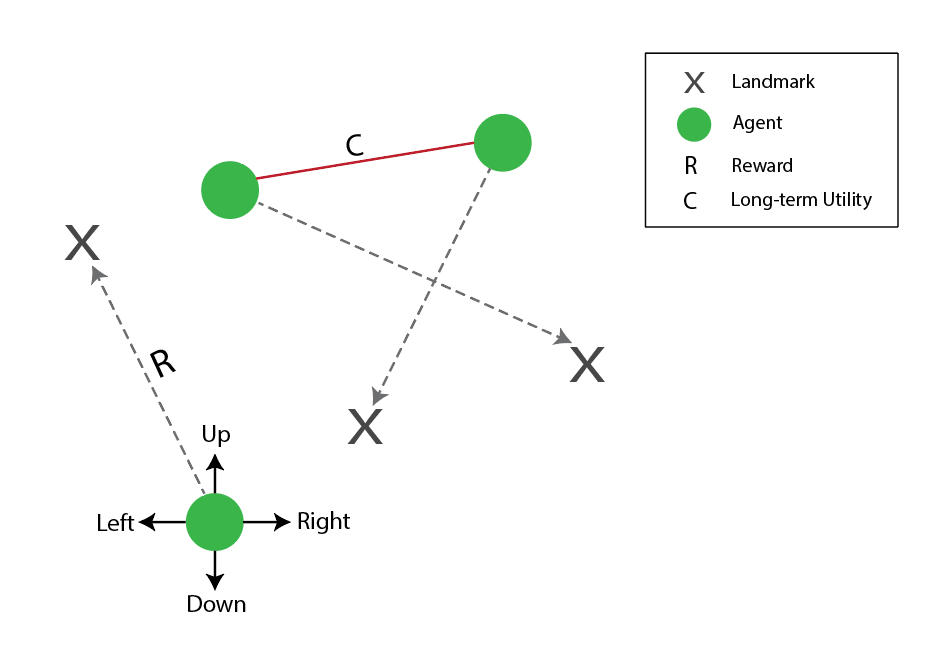}
    
    \caption{Pictorial representation of the cooperative navigation environment}
    \label{fig:coopnav_env}
\end{figure}
\begin{figure}[h]
    \centering
    \includegraphics[width=0.65\textwidth]{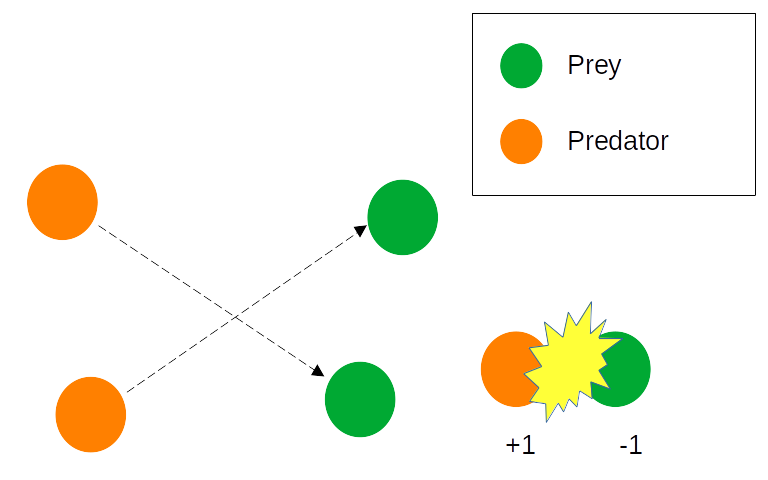}
    
    \caption{Pictorial representation of the predator-prey environment }
    \label{fig:predprey_env}
\end{figure}

\subsubsection{Implementation}
\label{sec:implementation}
We use the PyTorch framework in Python to evaluate our algorithm outlined in Section~\ref{sec:algorithm}. Gradient estimation of the neural networks for the policy and value networks is done using PyTorch's automatic differentiation engine. The test environment is built using the OpenAI Gym framework. The experiments were run on a system with NVIDIA Titan Xp 12 GB GPU, AMD Ryzen 5 3600 6-Core Processor, and 32 GB of RAM. The hyper-parameters chosen for our implementation are detailed below:

\paragraph{Parameters}
\label{ssec:parameters}
Each agent's policy is parameterized by a neural network, where there are two hidden layers, with the first one having 128 neurons and the second having 64 neurons. The final layer is a softmax layer. The dimension of the input layer is 20, and the output layer’s is 5, to match the dimensions of the observation space and the action space, respectively. We also use a critic network as our baseline, with two hidden layers similar to the policy network. The input and output dimensions of the critic network are 20 and 1. The discount factor $\gamma$ for all tests is 0.99. The learning rates $\eta_1$, $\eta_2$ and the critic learning rate are $0.0003$, $0.001$ and $0.001$ respectively. $\eta_1$ is set lower than the other two learning rates so that the descent induced by it does not overpower the ascent induced by $\eta_2$. The $\beta$ parameter is set to $0.2$. $c_i$ for all agents is set to 10.  All experiments were conducted using 5 random seeds each, averaging the results to get the final graph. 

\subsubsection{Experiments}
\label{sec:experiments}

We have run several experiments, to evaluate the performance and efficiency of the algorithm as well as compare it which can be divided into three parts. The first part includes experiments to analyze how well the algorithm performs in different topology graphs and with varying numbers of agents. In particular, six experiments were run for each environment.
For the cooperative navigation environment, the experiments were conducted with $n=3$, $4$, and $5$ agents, respectively, where the agents were connected using one of three different network topologies, namely ring, fully-connected, and bi-partite (as shown in Figure~\ref{fig:graphs}). The results of these experiments are illustrated in Figure~\ref{fig:exp_test_coopnav}. For the predator-prey environment, similar experiments were done with the following settings: 1 predator 1 prey, 2 predators 1 prey, and 3 predators 2 prey. The results are depicted in Figure~\ref{fig:exp_test_predprey}.

In the second part, we systematically evaluated the performance of our algorithm DePaint in comparison to the established CMIX algorithm across varying agent populations. Our investigation focused on the algorithm's convergence behavior and efficiency as the number of agents increased while preserving consistent initial agent configurations. We run both algorithms with $3$, $4$, and $5$ agents connected in a ring topology on both environments. These results are illustrated in Figure~\ref{fig:exp_compare}.

For the final part, we conducted an ablation study to show the efficacy of the momentum-based variance reduction. Here we compared two versions of our algorithm, one with momentum-based variance reduction and one without. The experiment was conducted in the co-operative navigation environment with $n=5$ agents connected using the ring topology. The findings are shown in Figure~\ref{fig:exp_ablation}.

\begin{figure}[h]
    \centering
    \includegraphics[width=0.8\textwidth]{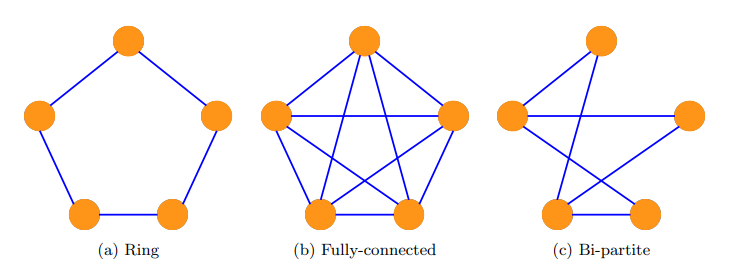}
    \caption{The three network topologies tested in Experiment~\ref{exp:exp_test}.}
    \label{fig:graphs}
\end{figure}

\subsection{Empirical Results}
\label{sec:results}
This section shows the experimental results from the experiments designed in Section~\ref{sec:experiments}. It also includes the result analysis for each of the graphs.

\refstepcounter{experiment}
\label{exp:exp_test}
\subsubsection{Experiment \arabic{experiment}: Performance and Convergence Analysis}
\label{sec:exp_test}
\begin{figure}[H]
    \centering
    \includegraphics[width=\textwidth]{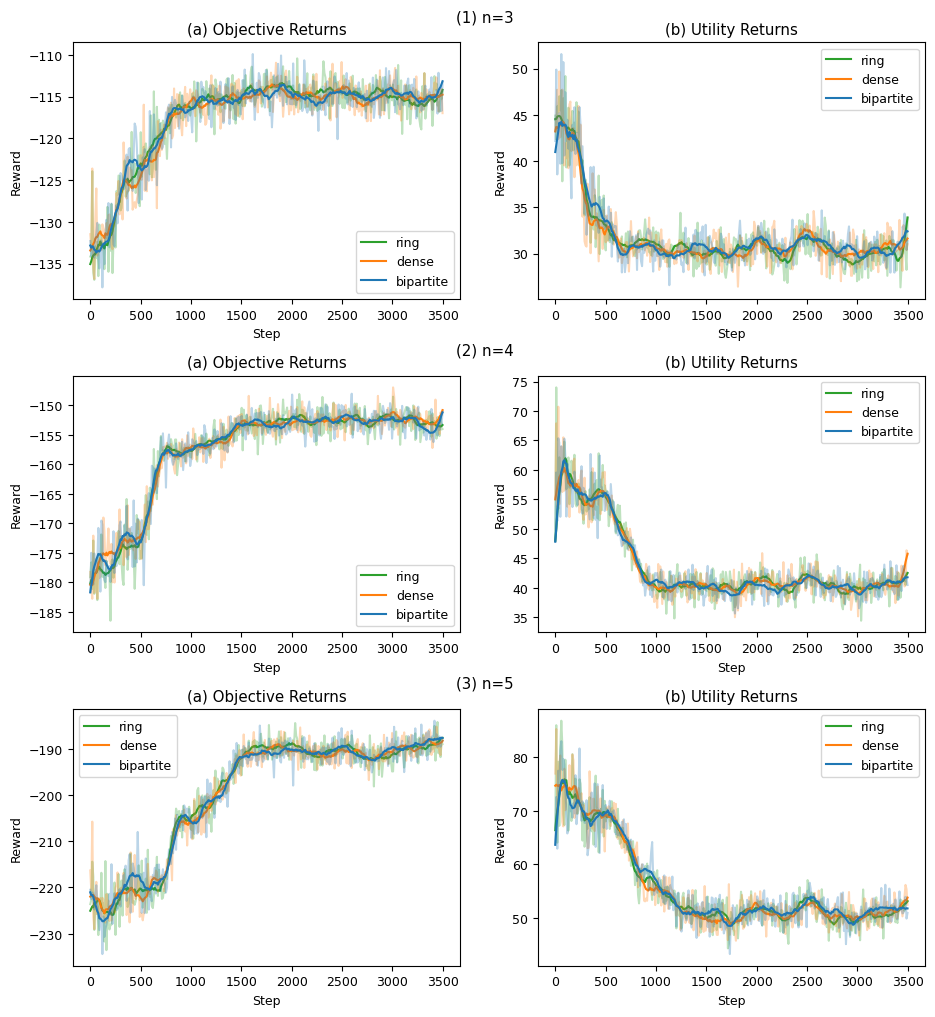}
    \caption{Objective return and utility returns of DePAint in different settings. The experiment was done with n=3, 4, and 5 agents in three different strongly connected graphs: Bipartite, Dense, Ring.}
    \label{fig:exp_test_coopnav}
\end{figure}
\begin{figure}[H]
    \centering
    \includegraphics[width=\textwidth]{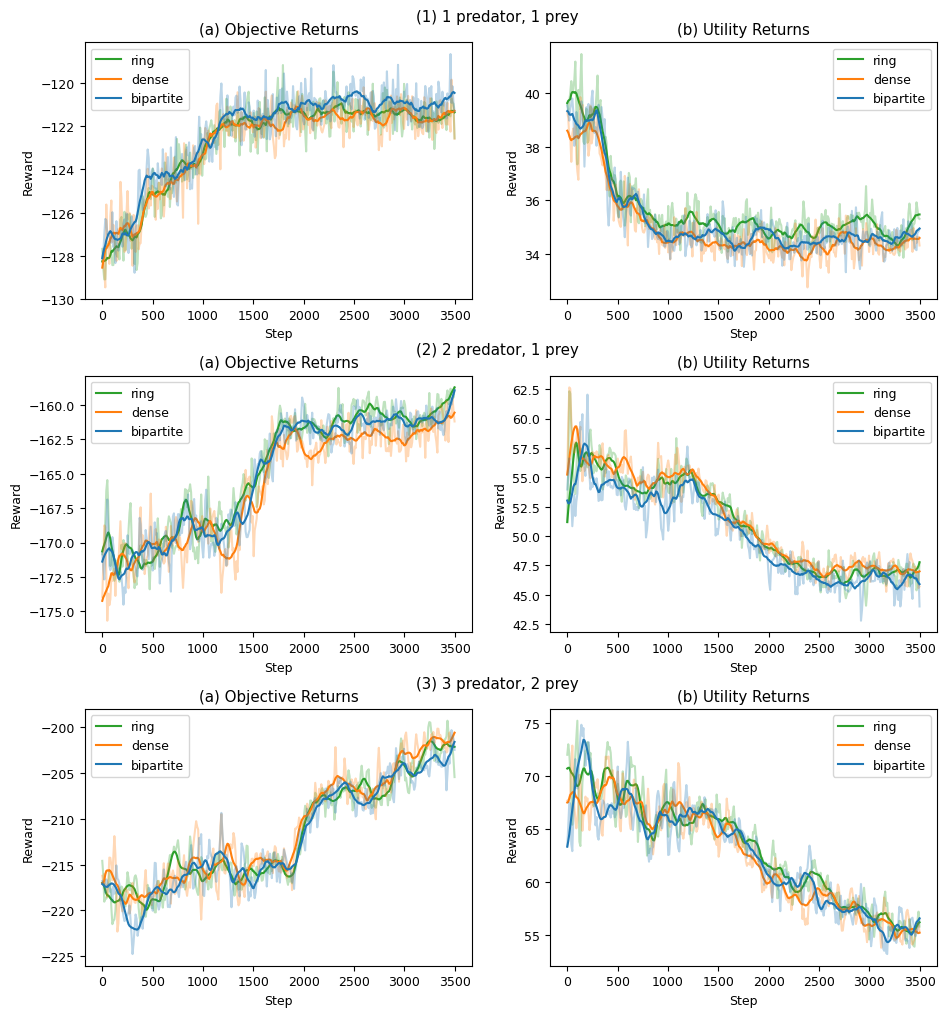}
    
    \caption{Objective return and utility returns of DePAint in different settings. The experiment was done with 1 predator 1 prey, 2 predators 1 prey, and 3 predators 2 preys in three different strongly connected graphs: Bipartite, Dense, Ring.}
    \label{fig:exp_test_predprey}
\end{figure}

\refstepcounter{experiment}
\label{exp:exp_compare}
\subsubsection{Experiment \arabic{experiment}: Comparative and Scalability Analysis}
\label{sec:exp_compare}
\begin{figure}[H]
    \centering
    \includegraphics[width=\textwidth]{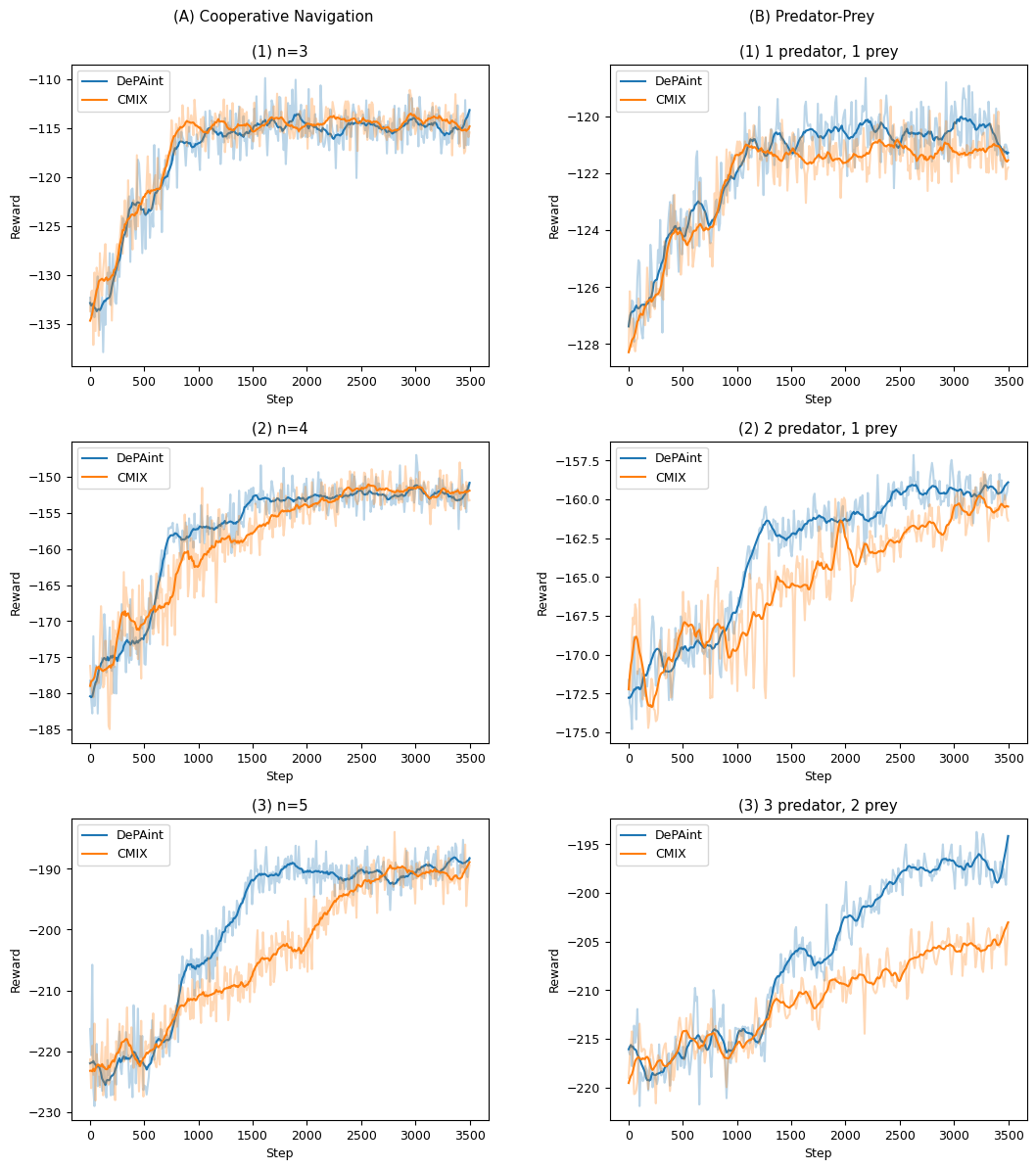}
    
    \caption{Objective Return graph for DePAint and CMIX, in the cooperative navigation and predator-prey environment. The experiment was carried out with n = 3, 4 and 5 agents for the cooperative navigation environment. As for the predator-prey environment the following settings were used: 1 predator 1 prey, 2 predators 1 prey, and 3 predators 2 prey. The ring topology used for DePAint}
    \label{fig:exp_compare}
\end{figure}

\refstepcounter{experiment}
\label{exp:exp_compare}
\subsubsection{Experiment \arabic{experiment}: Ablation Study}
\label{sec:exp_compare}
\begin{figure}[H]
    \centering
    \includegraphics[width=\textwidth]{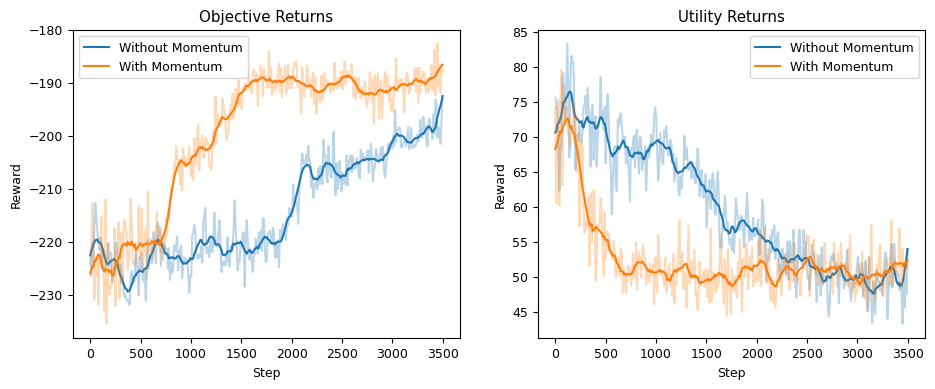}
    
    \caption{Objective Return graph for DePAint with and without momentum-based variance reduction. The experiment was carried out in the cooperative navigation environment with n = 5 agents connected using ring topology}
    \label{fig:exp_ablation}
\end{figure}

\subsubsection{Summary of Experimental Results}
\label{experimentsummary}

The line graphs displayed in Figure \ref{fig:exp_test_coopnav} depict the average sum of objective (Figure \ref{fig:exp_test_coopnav}a) and utility (Figure \ref{fig:exp_test_coopnav}b) returns in the cooperative navigation environment. The experiments were conducted with three distinct topologies (Ring, Dense, and Bipartite) and varying numbers of agents ($n = 3, n = 4, n = 5$). Across these settings, both the objective return and utility return demonstrated consistent trends, with the objective return gradually increasing up to 1000 steps for $n = 3$ and $n = 4$, and up to 1500 steps for $n = 5$. Similarly, the utility return for $n = 3$ declined up to 500 steps, and for $n = 4$ and $n = 5$, it gradually declined up to 1000 steps. These patterns remained comparable across the different topological configurations, confirming the algorithm's resilience to variations in communication network structure in the cooperative navigation environment.

In Figure~\ref{fig:exp_test_predprey}, the line graphs represent the corresponding results for the predator-prey environment. The experiments were conducted with varying compositions of agents: 1 predator 1 prey, 2 predators 1 prey, and 3 predators 2 prey. Similar to the cooperative navigation environment, the objective and utility returns displayed analogous trends. However, convergence in the predator-prey environment occurred at a slower rate compared to cooperative navigation. This deceleration in convergence is attributed to the heterogeneous nature of agents in the predator-prey scenario, introducing additional complexity to the environment. 

Figure \ref{fig:exp_compare} shows the comparison of the objective return of both DePAint and CMIX. It was tested with three different numbers of agents ($n = 3, n = 4, n = 5$), and the ring topology was used for DePAint, which represents a minimal communication network between the agents. Overall, both of them perform pretty well, but our algorithm DePAint outperforms CMIX when the agent number is $n=4$ and $n=5$. 
Due to CMIX being a centralized algorithm, the convergence rate suffers when the number of agents increases. On the other hand, DePAint can overcome this added complexity due to its speedup on account of having a decentralized setup. This shows the superior scalability of DePAint compared to CMIX. Our algorithm performs better when we expand the number of agents, which is to be expected for a decentralized algorithm compared to a centralized one.

Finally, Figure~\ref{fig:exp_ablation} shows the effectiveness of the momentum-based variance reduction technique on the convergence of the algorithm. As we see in the line graph, the objective rewards climb much faster and converge earlier in the version with momentum than in the version without. This shows the necessity of the momentum-based approach in speeding up the convergence of the algorithm.

\section{Conclusions and Future Work}
\label{sec:conclusion}
In this work, we have proposed an efficient algorithm that can solve decentralized multiagent CMDP problems where each agent must maximize the long-term discounted return subject to both peak and average constraints. Numerical results show that the algorithm is capable of performing well with a minimal communication network. Furthermore, according to our empirical evidence, it is more scalable than existing centralized algorithms that deal with similar constraints.  Additionally, it is pertinent to highlight the privacy-preserving nature of our algorithm. Since the rewards and constraints are inherently private to each agent, the decentralized approach ensures that sensitive information remains confidential. This aspect adds an extra layer of significance to our work, as privacy considerations become increasingly crucial in various multiagent systems.

Finally, it is worth noting the potential for algorithmic enhancements. Specifically, the incorporation of natural gradient descent \cite{zhang2019fast} stands as a promising avenue for future refinement. Exploring other variance reduction techniques could contribute to improving the algorithm's stability, while the introduction of asynchronous learning may enhance its practicality, making it better suited for real-world scenarios. Moreover, further work can be done to improve upon the algorithm's performance in environments offering sparse and delayed rewards. This study serves as a foundational step towards the development of more robust, safety-conscious, and privacy-preserving algorithms, setting the stage for continued advancements in the field.

\section{Data Availability and Access}
This research particularly does not involve any explicit datasets. The code used to implement the model described in this paper can be found in the GitHub repository "depaint", which can be accessed via the link \href{https://github.com/hoenchioma/depaint}{https://github.com/hoenchioma/depaint}. The code in this repository can be used to replicate the experimental results similar to those presented in this paper. 

\paragraph{Ethical and Informed Consent for Data Used}
In the research presented within this journal, we emphasize our commitment to ethical data practices.
In the context of our research, it is essential to address the unique data dynamics inherent to reinforcement learning and simulations. Unlike traditional datasets, our study revolves around agent interactions within simulated environments, and as such, ethical considerations differ from those concerning personal or sensitive data.

\section{Author Contribution Statement}
All authors contributed to the conception and design of the study. All authors approved the final manuscript. The authors confirm contribution to the paper as follows,
\begin{itemize}
    \item \textbf{Raheeb Hassan:} conceptualization, methodology, writing - original draft, experimentation
    \item \textbf{K.M. Shadman Wadith:} conceptualization, methodology, theoretical analysis, writing- original draft.
    \item \textbf{Md. Mamun or Rashid:}  editing the draft, visualization and investigation.
    \item \textbf{Md. Mosaddek Khan:} conceptualization, supervision, project administration, finalizing the draft.
\end{itemize}

\paragraph{Competing Interest}
The authors declare that they have no known competing financial interests or personal relationships that could have appeared to influence the work reported in this paper.

\bibliography{bibliography}


\begin{thebibliography}{32}
\ifx \bisbn   \undefined \def \bisbn  #1{ISBN #1}\fi
\ifx \binits  \undefined \def \binits#1{#1}\fi
\ifx \bauthor  \undefined \def \bauthor#1{#1}\fi
\ifx \batitle  \undefined \def \batitle#1{#1}\fi
\ifx \bjtitle  \undefined \def \bjtitle#1{#1}\fi
\ifx \bvolume  \undefined \def \bvolume#1{\textbf{#1}}\fi
\ifx \byear  \undefined \def \byear#1{#1}\fi
\ifx \bissue  \undefined \def \bissue#1{#1}\fi
\ifx \bfpage  \undefined \def \bfpage#1{#1}\fi
\ifx \blpage  \undefined \def \blpage #1{#1}\fi
\ifx \burl  \undefined \def \burl#1{\textsf{#1}}\fi
\ifx \doiurl  \undefined \def \doiurl#1{\url{https://doi.org/#1}}\fi
\ifx \betal  \undefined \def \betal{\textit{et al.}}\fi
\ifx \binstitute  \undefined \def \binstitute#1{#1}\fi
\ifx \binstitutionaled  \undefined \def \binstitutionaled#1{#1}\fi
\ifx \bctitle  \undefined \def \bctitle#1{#1}\fi
\ifx \beditor  \undefined \def \beditor#1{#1}\fi
\ifx \bpublisher  \undefined \def \bpublisher#1{#1}\fi
\ifx \bbtitle  \undefined \def \bbtitle#1{#1}\fi
\ifx \bedition  \undefined \def \bedition#1{#1}\fi
\ifx \bseriesno  \undefined \def \bseriesno#1{#1}\fi
\ifx \blocation  \undefined \def \blocation#1{#1}\fi
\ifx \bsertitle  \undefined \def \bsertitle#1{#1}\fi
\ifx \bsnm \undefined \def \bsnm#1{#1}\fi
\ifx \bsuffix \undefined \def \bsuffix#1{#1}\fi
\ifx \bparticle \undefined \def \bparticle#1{#1}\fi
\ifx \barticle \undefined \def \barticle#1{#1}\fi
\bibcommenthead
\ifx \bconfdate \undefined \def \bconfdate #1{#1}\fi
\ifx \botherref \undefined \def \botherref #1{#1}\fi
\ifx \url \undefined \def \url#1{\textsf{#1}}\fi
\ifx \bchapter \undefined \def \bchapter#1{#1}\fi
\ifx \bbook \undefined \def \bbook#1{#1}\fi
\ifx \bcomment \undefined \def \bcomment#1{#1}\fi
\ifx \oauthor \undefined \def \oauthor#1{#1}\fi
\ifx \citeauthoryear \undefined \def \citeauthoryear#1{#1}\fi
\ifx \endbibitem  \undefined \def \endbibitem {}\fi
\ifx \bconflocation  \undefined \def \bconflocation#1{#1}\fi
\ifx \arxivurl  \undefined \def \arxivurl#1{\textsf{#1}}\fi
\csname PreBibitemsHook\endcsname

\bibitem[\protect\citeauthoryear{Amodei et~al.}{2016}]{amodei2016concrete}
\begin{botherref}
\oauthor{\bsnm{Amodei}, \binits{D.}},
\oauthor{\bsnm{Olah}, \binits{C.}},
\oauthor{\bsnm{Steinhardt}, \binits{J.}},
\oauthor{\bsnm{Christiano}, \binits{P.}},
\oauthor{\bsnm{Schulman}, \binits{J.}},
\oauthor{\bsnm{Man{\'e}}, \binits{D.}}:
Concrete problems in ai safety.
arXiv preprint arXiv:1606.06565
(2016)
\end{botherref}
\endbibitem

\bibitem[\protect\citeauthoryear{Shalev-Shwartz et~al.}{2016}]{shalev2016safe}
\begin{botherref}
\oauthor{\bsnm{Shalev-Shwartz}, \binits{S.}},
\oauthor{\bsnm{Shammah}, \binits{S.}},
\oauthor{\bsnm{Shashua}, \binits{A.}}:
Safe, multi-agent, reinforcement learning for autonomous driving.
arXiv preprint arXiv:1610.03295
(2016)
\end{botherref}
\endbibitem

\bibitem[\protect\citeauthoryear{Alqahtani et~al.}{2022}]{alqahtani2022dynamic}
\begin{barticle}
\bauthor{\bsnm{Alqahtani}, \binits{M.}},
\bauthor{\bsnm{Scott}, \binits{M.J.}},
\bauthor{\bsnm{Hu}, \binits{M.}}:
\batitle{Dynamic energy scheduling and routing of a large fleet of electric vehicles using multi-agent reinforcement learning}.
\bjtitle{Computers \& Industrial Engineering}
\bvolume{169},
\bfpage{108180}
(\byear{2022})
\end{barticle}
\endbibitem

\bibitem[\protect\citeauthoryear{Altman}{1995}]{altman1995constrained}
\begin{botherref}
\oauthor{\bsnm{Altman}, \binits{E.}}:
Constrained markov decision processes.
PhD thesis,
INRIA
(1995)
\end{botherref}
\endbibitem

\bibitem[\protect\citeauthoryear{Achiam et~al.}{2017}]{achiam2017constrained}
\begin{bchapter}
\bauthor{\bsnm{Achiam}, \binits{J.}},
\bauthor{\bsnm{Held}, \binits{D.}},
\bauthor{\bsnm{Tamar}, \binits{A.}},
\bauthor{\bsnm{Abbeel}, \binits{P.}}:
\bctitle{Constrained policy optimization}.
In: \bbtitle{International Conference on Machine Learning},
pp. \bfpage{22}--\blpage{31}
(\byear{2017}).
\bcomment{PMLR}
\end{bchapter}
\endbibitem

\bibitem[\protect\citeauthoryear{Gu et~al.}{2021}]{gu2021multi}
\begin{botherref}
\oauthor{\bsnm{Gu}, \binits{S.}},
\oauthor{\bsnm{Kuba}, \binits{J.G.}},
\oauthor{\bsnm{Wen}, \binits{M.}},
\oauthor{\bsnm{Chen}, \binits{R.}},
\oauthor{\bsnm{Wang}, \binits{Z.}},
\oauthor{\bsnm{Tian}, \binits{Z.}},
\oauthor{\bsnm{Wang}, \binits{J.}},
\oauthor{\bsnm{Knoll}, \binits{A.}},
\oauthor{\bsnm{Yang}, \binits{Y.}}:
Multi-agent constrained policy optimisation.
arXiv preprint arXiv:2110.02793
(2021)
\end{botherref}
\endbibitem

\bibitem[\protect\citeauthoryear{Gronauer and Diepold}{2021}]{gronauer2021multi}
\begin{botherref}
\oauthor{\bsnm{Gronauer}, \binits{S.}},
\oauthor{\bsnm{Diepold}, \binits{K.}}:
Multi-agent deep reinforcement learning: a survey.
Artificial Intelligence Review,
1--49
(2021)
\end{botherref}
\endbibitem

\bibitem[\protect\citeauthoryear{Lowe et~al.}{2017}]{lowe2017multi}
\begin{botherref}
\oauthor{\bsnm{Lowe}, \binits{R.}},
\oauthor{\bsnm{Wu}, \binits{Y.I.}},
\oauthor{\bsnm{Tamar}, \binits{A.}},
\oauthor{\bsnm{Harb}, \binits{J.}},
\oauthor{\bsnm{Pieter~Abbeel}, \binits{O.}},
\oauthor{\bsnm{Mordatch}, \binits{I.}}:
Multi-agent actor-critic for mixed cooperative-competitive environments.
Advances in neural information processing systems
\textbf{30}
(2017)
\end{botherref}
\endbibitem

\bibitem[\protect\citeauthoryear{Parnika et~al.}{2021}]{parnika2021attention}
\begin{botherref}
\oauthor{\bsnm{Parnika}, \binits{P.}},
\oauthor{\bsnm{Diddigi}, \binits{R.B.}},
\oauthor{\bsnm{Danda}, \binits{S.K.R.}},
\oauthor{\bsnm{Bhatnagar}, \binits{S.}}:
Attention Actor-Critic algorithm for Multi-Agent Constrained Co-operative Reinforcement Learning
(2021)
\end{botherref}
\endbibitem

\bibitem[\protect\citeauthoryear{Lu et~al.}{2021}]{lu2021decentralized}
\begin{bchapter}
\bauthor{\bsnm{Lu}, \binits{S.}},
\bauthor{\bsnm{Zhang}, \binits{K.}},
\bauthor{\bsnm{Chen}, \binits{T.}},
\bauthor{\bsnm{Basar}, \binits{T.}},
\bauthor{\bsnm{Horesh}, \binits{L.}}:
\bctitle{Decentralized policy gradient descent ascent for safe multi-agent reinforcement learning}.
In: \bbtitle{Proceedings of the AAAI Conference on Artificial Intelligence},
vol. \bseriesno{35},
pp. \bfpage{8767}--\blpage{8775}
(\byear{2021})
\end{bchapter}
\endbibitem

\bibitem[\protect\citeauthoryear{Bai et~al.}{2020}]{bai2020provably}
\begin{botherref}
\oauthor{\bsnm{Bai}, \binits{Q.}},
\oauthor{\bsnm{Aggarwal}, \binits{V.}},
\oauthor{\bsnm{Gattami}, \binits{A.}}:
Provably efficient model-free algorithm for mdps with peak constraints.
arXiv preprint arXiv:2003.05555
(2020)
\end{botherref}
\endbibitem

\bibitem[\protect\citeauthoryear{Gattami}{2019}]{gattami2019reinforcement}
\begin{botherref}
\oauthor{\bsnm{Gattami}, \binits{A.}}:
Reinforcement learning of markov decision processes with peak constraints.
arXiv preprint arXiv:1901.07839
(2019)
\end{botherref}
\endbibitem

\bibitem[\protect\citeauthoryear{Geibel}{2006}]{geibel2006reinforcement}
\begin{bchapter}
\bauthor{\bsnm{Geibel}, \binits{P.}}:
\bctitle{Reinforcement learning for mdps with constraints}.
In: \bbtitle{European Conference on Machine Learning},
pp. \bfpage{646}--\blpage{653}
(\byear{2006}).
\bcomment{Springer}
\end{bchapter}
\endbibitem

\bibitem[\protect\citeauthoryear{Geibel and Wysotzki}{2005}]{geibel2005risk}
\begin{barticle}
\bauthor{\bsnm{Geibel}, \binits{P.}},
\bauthor{\bsnm{Wysotzki}, \binits{F.}}:
\batitle{Risk-sensitive reinforcement learning applied to control under constraints}.
\bjtitle{Journal of Artificial Intelligence Research}
\bvolume{24},
\bfpage{81}--\blpage{108}
(\byear{2005})
\end{barticle}
\endbibitem

\bibitem[\protect\citeauthoryear{Chow et~al.}{2018}]{chow2018lyapunov}
\begin{botherref}
\oauthor{\bsnm{Chow}, \binits{Y.}},
\oauthor{\bsnm{Nachum}, \binits{O.}},
\oauthor{\bsnm{Duenez-Guzman}, \binits{E.}},
\oauthor{\bsnm{Ghavamzadeh}, \binits{M.}}:
A lyapunov-based approach to safe reinforcement learning.
Advances in neural information processing systems
\textbf{31}
(2018)
\end{botherref}
\endbibitem

\bibitem[\protect\citeauthoryear{Ding et~al.}{2021}]{ding2021provably}
\begin{bchapter}
\bauthor{\bsnm{Ding}, \binits{D.}},
\bauthor{\bsnm{Wei}, \binits{X.}},
\bauthor{\bsnm{Yang}, \binits{Z.}},
\bauthor{\bsnm{Wang}, \binits{Z.}},
\bauthor{\bsnm{Jovanovic}, \binits{M.}}:
\bctitle{Provably efficient safe exploration via primal-dual policy optimization}.
In: \bbtitle{International Conference on Artificial Intelligence and Statistics},
pp. \bfpage{3304}--\blpage{3312}
(\byear{2021}).
\bcomment{PMLR}
\end{bchapter}
\endbibitem

\bibitem[\protect\citeauthoryear{Prashanth and Ghavamzadeh}{2016}]{prashanth2016variance}
\begin{barticle}
\bauthor{\bsnm{Prashanth}, \binits{L.}},
\bauthor{\bsnm{Ghavamzadeh}, \binits{M.}}:
\batitle{Variance-constrained actor-critic algorithms for discounted and average reward mdps}.
\bjtitle{Machine Learning}
\bvolume{105}(\bissue{3}),
\bfpage{367}--\blpage{417}
(\byear{2016})
\end{barticle}
\endbibitem

\bibitem[\protect\citeauthoryear{Liu et~al.}{2021}]{liu2021cmix}
\begin{bchapter}
\bauthor{\bsnm{Liu}, \binits{C.}},
\bauthor{\bsnm{Geng}, \binits{N.}},
\bauthor{\bsnm{Aggarwal}, \binits{V.}},
\bauthor{\bsnm{Lan}, \binits{T.}},
\bauthor{\bsnm{Yang}, \binits{Y.}},
\bauthor{\bsnm{Xu}, \binits{M.}}:
\bctitle{Cmix: Deep multi-agent reinforcement learning with peak and average constraints}.
In: \bbtitle{Joint European Conference on Machine Learning and Knowledge Discovery in Databases},
pp. \bfpage{157}--\blpage{173}
(\byear{2021}).
\bcomment{Springer}
\end{bchapter}
\endbibitem

\bibitem[\protect\citeauthoryear{Rashid et~al.}{2018}]{rashid2018qmix}
\begin{bchapter}
\bauthor{\bsnm{Rashid}, \binits{T.}},
\bauthor{\bsnm{Samvelyan}, \binits{M.}},
\bauthor{\bsnm{Schroeder}, \binits{C.}},
\bauthor{\bsnm{Farquhar}, \binits{G.}},
\bauthor{\bsnm{Foerster}, \binits{J.}},
\bauthor{\bsnm{Whiteson}, \binits{S.}}:
\bctitle{Qmix: Monotonic value function factorisation for deep multi-agent reinforcement learning}.
In: \bbtitle{International Conference on Machine Learning},
pp. \bfpage{4295}--\blpage{4304}
(\byear{2018}).
\bcomment{PMLR}
\end{bchapter}
\endbibitem

\bibitem[\protect\citeauthoryear{Geng et~al.}{2023}]{geng2023reinforcement}
\begin{botherref}
\oauthor{\bsnm{Geng}, \binits{N.}},
\oauthor{\bsnm{Bai}, \binits{Q.}},
\oauthor{\bsnm{Liu}, \binits{C.}},
\oauthor{\bsnm{Lan}, \binits{T.}},
\oauthor{\bsnm{Aggarwal}, \binits{V.}},
\oauthor{\bsnm{Yang}, \binits{Y.}},
\oauthor{\bsnm{Xu}, \binits{M.}}:
A reinforcement learning framework for vehicular network routing under peak and average constraints.
IEEE Transactions on Vehicular Technology
(2023)
\end{botherref}
\endbibitem

\bibitem[\protect\citeauthoryear{Watkins and Dayan}{1992}]{watkins1992q}
\begin{barticle}
\bauthor{\bsnm{Watkins}, \binits{C.J.}},
\bauthor{\bsnm{Dayan}, \binits{P.}}:
\batitle{Q-learning}.
\bjtitle{Machine learning}
\bvolume{8}(\bissue{3}),
\bfpage{279}--\blpage{292}
(\byear{1992})
\end{barticle}
\endbibitem

\bibitem[\protect\citeauthoryear{Rummery and Niranjan}{1994}]{rummery1994line}
\begin{bbook}
\bauthor{\bsnm{Rummery}, \binits{G.A.}},
\bauthor{\bsnm{Niranjan}, \binits{M.}}:
\bbtitle{On-line Q-learning Using Connectionist Systems}
vol. \bseriesno{37}.
\bpublisher{University of Cambridge, Department of Engineering Cambridge},
\blocation{UK}
(\byear{1994})
\end{bbook}
\endbibitem

\bibitem[\protect\citeauthoryear{Bertsekas}{2014}]{bertsekas2014constrained}
\begin{bbook}
\bauthor{\bsnm{Bertsekas}, \binits{D.P.}}:
\bbtitle{Constrained Optimization and Lagrange Multiplier Methods}.
\bpublisher{Academic press (Massachusetts Institute of Technology)},
\blocation{USA}
(\byear{2014})
\end{bbook}
\endbibitem

\bibitem[\protect\citeauthoryear{Beznosikov et~al.}{2023}]{beznosikov2023stochastic}
\begin{bchapter}
\bauthor{\bsnm{Beznosikov}, \binits{A.}},
\bauthor{\bsnm{Gorbunov}, \binits{E.}},
\bauthor{\bsnm{Berard}, \binits{H.}},
\bauthor{\bsnm{Loizou}, \binits{N.}}:
\bctitle{Stochastic gradient descent-ascent: Unified theory and new efficient methods}.
In: \bbtitle{International Conference on Artificial Intelligence and Statistics},
pp. \bfpage{172}--\blpage{235}
(\byear{2023}).
\bcomment{PMLR}
\end{bchapter}
\endbibitem

\bibitem[\protect\citeauthoryear{Xian et~al.}{2021}]{xian2021faster}
\begin{barticle}
\bauthor{\bsnm{Xian}, \binits{W.}},
\bauthor{\bsnm{Huang}, \binits{F.}},
\bauthor{\bsnm{Zhang}, \binits{Y.}},
\bauthor{\bsnm{Huang}, \binits{H.}}:
\batitle{A faster decentralized algorithm for nonconvex minimax problems}.
\bjtitle{Advances in Neural Information Processing Systems}
\bvolume{34},
\bfpage{25865}--\blpage{25877}
(\byear{2021})
\end{barticle}
\endbibitem

\bibitem[\protect\citeauthoryear{Li et~al.}{2020}]{li2020communication}
\begin{bchapter}
\bauthor{\bsnm{Li}, \binits{B.}},
\bauthor{\bsnm{Cen}, \binits{S.}},
\bauthor{\bsnm{Chen}, \binits{Y.}},
\bauthor{\bsnm{Chi}, \binits{Y.}}:
\bctitle{Communication-efficient distributed optimization in networks with gradient tracking and variance reduction}.
In: \bbtitle{International Conference on Artificial Intelligence and Statistics},
pp. \bfpage{1662}--\blpage{1672}
(\byear{2020}).
\bcomment{PMLR}
\end{bchapter}
\endbibitem

\bibitem[\protect\citeauthoryear{Pu and Nedi{\'c}}{2021}]{pu2021distributed}
\begin{barticle}
\bauthor{\bsnm{Pu}, \binits{S.}},
\bauthor{\bsnm{Nedi{\'c}}, \binits{A.}}:
\batitle{Distributed stochastic gradient tracking methods}.
\bjtitle{Mathematical Programming}
\bvolume{187}(\bissue{1}),
\bfpage{409}--\blpage{457}
(\byear{2021})
\end{barticle}
\endbibitem

\bibitem[\protect\citeauthoryear{Cutkosky and Orabona}{2019}]{cutkosky2019momentum}
\begin{botherref}
\oauthor{\bsnm{Cutkosky}, \binits{A.}},
\oauthor{\bsnm{Orabona}, \binits{F.}}:
Momentum-based variance reduction in non-convex sgd.
Advances in neural information processing systems
\textbf{32}
(2019)
\end{botherref}
\endbibitem

\bibitem[\protect\citeauthoryear{Tran-Dinh et~al.}{2019}]{tran2019hybrid}
\begin{botherref}
\oauthor{\bsnm{Tran-Dinh}, \binits{Q.}},
\oauthor{\bsnm{Pham}, \binits{N.H.}},
\oauthor{\bsnm{Phan}, \binits{D.T.}},
\oauthor{\bsnm{Nguyen}, \binits{L.M.}}:
Hybrid stochastic gradient descent algorithms for stochastic nonconvex optimization.
arXiv preprint arXiv:1905.05920
(2019)
\end{botherref}
\endbibitem

\bibitem[\protect\citeauthoryear{Jiang et~al.}{2022}]{jiang2022mdpgt}
\begin{bchapter}
\bauthor{\bsnm{Jiang}, \binits{Z.}},
\bauthor{\bsnm{Lee}, \binits{X.Y.}},
\bauthor{\bsnm{Tan}, \binits{S.Y.}},
\bauthor{\bsnm{Tan}, \binits{K.L.}},
\bauthor{\bsnm{Balu}, \binits{A.}},
\bauthor{\bsnm{Lee}, \binits{Y.M.}},
\bauthor{\bsnm{Hegde}, \binits{C.}},
\bauthor{\bsnm{Sarkar}, \binits{S.}}:
\bctitle{Mdpgt: momentum-based decentralized policy gradient tracking}.
In: \bbtitle{Proceedings of the AAAI Conference on Artificial Intelligence},
vol. \bseriesno{36},
pp. \bfpage{9377}--\blpage{9385}
(\byear{2022})
\end{bchapter}
\endbibitem

\bibitem[\protect\citeauthoryear{Zhang et~al.}{2018}]{zhang2018fully}
\begin{bchapter}
\bauthor{\bsnm{Zhang}, \binits{K.}},
\bauthor{\bsnm{Yang}, \binits{Z.}},
\bauthor{\bsnm{Liu}, \binits{H.}},
\bauthor{\bsnm{Zhang}, \binits{T.}},
\bauthor{\bsnm{Basar}, \binits{T.}}:
\bctitle{Fully decentralized multi-agent reinforcement learning with networked agents}.
In: \bbtitle{International Conference on Machine Learning},
pp. \bfpage{5872}--\blpage{5881}
(\byear{2018}).
\bcomment{PMLR}
\end{bchapter}
\endbibitem

\bibitem[\protect\citeauthoryear{Zhang et~al.}{2019}]{zhang2019fast}
\begin{botherref}
\oauthor{\bsnm{Zhang}, \binits{G.}},
\oauthor{\bsnm{Martens}, \binits{J.}},
\oauthor{\bsnm{Grosse}, \binits{R.B.}}:
Fast convergence of natural gradient descent for over-parameterized neural networks.
Advances in Neural Information Processing Systems
\textbf{32}
(2019)
\end{botherref}
\endbibitem

\end{thebibliography}

\end{document}